\documentclass[fleqn,orivec,envcountsect]{llncs}
\usepackage{sbrce}

\title{Process Algebra with Conditionals\\ in the Presence of Epsilon}

\author{J.A. Bergstra \and C.A. Middelburg}
\institute{Section Theory of Computation, Informatics Institute,
           University of Amsterdam \\
           Science Park~904, 1098~XH Amsterdam, the Netherlands \\
           \email{J.A.Bergstra@uva.nl, C.A.Middelburg@uva.nl}
          }

\begin{document}
\maketitle

\begin{abstract}
In a previous paper, we presented several extensions of \ACP\ with
conditional expressions, including one with a retrospection operator on
conditions to allow for looking back on conditions under which preceding
actions have been performed.
In this paper, we add a constant for a process that is only capable of
terminating successfully to those extensions of \ACP, which can be very
useful in applications.
It happens that in all cases the addition of this constant is
unproblematic.
\\[1.5ex]
\\[1.5ex]
{\sl 1998 CR Categories:}
D.1.3, D.2.1, D.2.4, F.1.2, F.3.1, F.3.2.
\end{abstract}

\section{Introduction}
\label{sect-intro}

In~\cite{BM05a}, we presented several extensions of
\ACP~\cite{BK84b,BW90} with conditional expressions.
The main extensions of \ACP\ presented in~\cite{BM05a} are \ACPc,
an extension of \ACP\ with conditional expressions of the form
$\zeta \gc p$ in which the conditions are taken from a free Boolean
algebra over a set of generators, \ACPcs, an extension \ACPc\ with a
signal emission operator on processes, and \ACPcr, an extension of
\ACPc\ with a retrospection operator on conditions.
Signal emission is usable for a special kind of condition evaluation.
Retrospection allows for looking back on conditions under which
preceding actions have been performed.
We also extended \ACPc\ and \ACPcr\ with operators devised for condition
evaluation and we outlined an application of \ACPcr\ in which it allows
for using conditions which express that a certain number of steps ago a
certain action must have been performed.


In this paper, a constant for a process that is only capable of
terminating successfully is added to the different extensions of \ACP\
presented in~\cite{BM05a}.
This constant is often referred to as the empty process constant.
In the past, the addition of the empty process constant to \ACP\ has
been treated in several ways.
The treatment in~\cite{KV85a} yields a non-associative parallel
composition operator.
The first treatment that yields an associative parallel composition
operator~\cite{Vra97a} is from 1986, but was not published until 1997.
The addition of the empty process constant to different extensions of
\ACP\ in this paper is based on~\cite{BG87c}.

It is clear from early work~\cite{KV85a,Vra97a} that the addition of the
empty process constant to \ACP\ was rather problematic.
Its addition to the different extensions of \ACP\ with conditional
expressions presented in~\cite{BM05a} turns out to present no additional
complications.
For that reason, we look upon this paper in its current form primarily
as supplementary material to~\cite{BM05a}.

The structure of this paper is as follows.
First of all, we introduce \ACPec, the extension of \ACPc\ with the
empty process constant (Section~\ref{sect-ACPec}).
After that, we introduce conditional transition systems and splitting
bisimilarity of conditional transition systems
(Section~\ref{sect-ts-bisim}) and the full splitting bisimulation
models of \ACPec, the main models of \ACPec\
(Section~\ref{sect-full-bisim-ACPec}).
Following this, we have a closer look at splitting bisimilarity based
on structural operational semantics
(Section~\ref{sect-sos-based-bisim}).
Next, we extend \ACPec\ with guarded recursion
(Section~\ref{sect-recursion}).
Thereupon, we extend \ACPec\ with condition evaluation operators
(Section~\ref{sect-cond-eval}), with state operators
(Section~\ref{sect-state-operator}) and with a signal emission operator
(Section~\ref{sect-signal-emission}); and analyse how those operators
are related.
We also adapt the full splitting bisimulation models of \ACPec\
to the full signal-observing splitting bisimulation models of \ACPecs,
the extension of \ACPec\ with signal emission
(Section~\ref{sect-full-bisim-ACPecs}).
After that, we extend \ACPec\ with a retrospection operator
(Section~\ref{sect-ACPecr}) and adapt the full splitting bisimulation
models of \ACPec\ to the full retrospective splitting bisimulation
models of \ACPecr, the extension of \ACPec\ with retrospection
(Section~\ref{sect-full-bisim-ACPecr}).
Thereupon, we extend \ACPecr\ with condition evaluation operators as
well (Section~\ref{sect-cond-eval-retro}).
We also outline an interesting application of \ACPecr\
(Section~\ref{sect-appl-retro}).
Finally, we make some concluding remarks
(Section~\ref{sect-conclusions}).

Some familiarity with Boolean algebras is desirable.
The definitions of all notions concerning Boolean algebras that are used
can be found in~\cite{MB89a}.

We thank Jan van Eijck.
He communicated an application of \ACPc\ to us which involves a register
update mechanism that cannot be dealt with in full generality without
the empty process constant.
This forms the greater part of our motivation to work out the addition
of the empty process constant to \ACPc.

\section{\ACPe\ with Conditions}
\label{sect-ACPec}

In this section, we present \ACPec, an extension of
\ACPe~\cite{BG87c,BW90} with conditional expressions of the form
$\zeta \gc p$.
\ACPec\ can be regarded as an extension of \ACPc~\cite{BM05a} with the
empty process constant too.
In \ACPec, as in \ACPe, it is assumed that a fixed but arbitrary finite
set of \emph{actions} $\Act$, with $\dead,\ep \not\in \Act$, and a fixed
but arbitrary commutative and associative \emph{communication} function
$\funct{\commm}{\Actd \x \Actd}{\Actd}$, such that
$\dead \commm a = \dead$ for all $a \in \Actd$, have been given.
The function $\commm$ is regarded to give the result of synchronously
performing any two actions for which this is possible, and to be $\dead$
otherwise.
Moreover, it is assumed that a fixed but arbitrary set of
\emph{atomic conditions} $\ACond$ has been given.

Let $\kappa$ be an infinite cardinal.
Then $\CondAlg{\kappa}$ is the free $\kappa$-complete Boolean algebra
over $\ACond$.%
\footnote
{For a definition of free $\kappa$-complete Boolean algebras, see
 e.g.~\cite{MB89a}.}
As usual, we identify Boolean algebras with their domain.
Thus, we also write $\CondAlg{\kappa}$ for the domain of
$\CondAlg{\kappa}$.
If $\kappa$ is regular,%
\footnote
{For a definition of regular cardinals, see e.g.~\cite{Sho67a,CK90a}.
 They include $\aleph_0$, $\aleph_1$, $\aleph_2$, \ldots\,.}
then $\CondAlg{\kappa}$ is isomorphic to the Boolean algebra of
equivalence classes with respect to logical equivalence of the set of
all propositions with elements of $\ACond$ as propositional variables
and with conjunctions and disjunctions of less than $\kappa$
propositions (see e.g.~\cite{MB89a}).
In \ACPec, conditions are taken from $\CondAlg{\aleph_0}$.
If $\ACond$ is a finite set, then
$\CondAlg{\kappa} = \CondAlg{\aleph_0}$ for all cardinals
$\kappa > \aleph_0$.
We are also interested in $\CondAlg{\kappa}$ for cardinals
$\kappa > \aleph_0$ because it permits us to consider infinitely
branching processes in the case where $\ACond$ is an infinite set.
Henceforth, we write $\FinCondAlg$ for $\CondAlg{\aleph_0}$.

The algebraic theory \ACPec\ has two sorts:
\begin{iteml}
\item
the sort $\Proc$ of \emph{processes};
\item
the sort $\Cond$ of \emph{conditions}.
\end{iteml}
The algebraic theory \ACPec\ has the following constants and
operators to  build terms of sort $\Proc$:
\begin{iteml}
\item
the \emph{deadlock} constant $\const{\dead}{\Proc}$;
\item
the \emph{empty process} constant $\const{\ep}{\Proc}$;
\item
for each $a \in \Act$, the \emph{action} constant $\const{a}{\Proc}$;
\item
the binary \emph{alternative composition} operator
$\funct{\altc}{\Proc \x \Proc}{\Proc}$;
\item
the binary \emph{sequential composition} operator
$\funct{\seqc}{\Proc \x \Proc}{\Proc}$;
\item
the binary \emph{guarded command} operator
$\funct{\gc}{\Cond \x \Proc}{\Proc}$;
\item
the binary \emph{parallel composition} operator
$\funct{\parc}{\Proc \x \Proc}{\Proc}$;
\item
the binary \emph{left merge} operator
$\funct{\leftm}{\Proc \x \Proc}{\Proc}$;
\item
the binary \emph{communication merge} operator
$\funct{\commm}{\Proc \x \Proc}{\Proc}$;
\item
for each $H \subseteq \Act$,
the unary \emph{encapsulation} operator
$\funct{\encap{H}}{\Proc}{\Proc}$.
\end{iteml}
The algebraic theory \ACPec\ has the following constants and
operators to build terms of sort $\Cond$:
\begin{iteml}
\item
the \emph{bottom} constant $\const{\bot}{\Cond}$;
\item
the \emph{top} constant $\const{\top}{\Cond}$;
\item
for each $\eta \in \ACond$, the \emph{atomic condition} constant
$\const{\eta}{\Cond}$;
\item
the unary \emph{complement} operator $\funct{-}{\Cond}{\Cond}$;
\item
the binary \emph{join} operator $\funct{\join}{\Cond \x \Cond}{\Cond}$;
\item
the binary \emph{meet} operator $\funct{\meet}{\Cond \x \Cond}{\Cond}$.
\end{iteml}
We use infix notation for the binary operators.
The following precedence conventions are used to reduce the need for
parentheses.
The operators to build terms of sort $\Cond$ bind stronger than the
operators to build terms of sort $\Proc$.
The operator ${} \seqc {}$ binds stronger than all other binary
operators to build terms of sort $\Proc$ and the operator ${} \altc {}$
binds weaker than all other binary operators to build terms of sort
$\Proc$.

The constants and operators of \ACPec\ to build terms of sort $\Proc$
are the constants and operators of \ACPe\ and additionally the guarded
command operator.
Let $p$ and $q$ be closed terms of sort $\Proc$ and $\zeta$ and $\xi$
be closed terms of sort $\Cond$, $a \in \Act$, $H \subseteq \Act$, and
$\eta \in \ACond$.
Then, intuitively, the constants and operators to build terms of sort
$\Proc$ can be explained as follows:
\begin{iteml}
\item
$\dead$ can neither perform an action nor terminate successfully;
\item
$\ep$ terminates successfully, unconditionally;
\item
$a$ first performs action $a$ and then terminates successfully, both
unconditionally;
\item
$p \altc q$ behaves either as $p$ or as $q$, but not both;
\item
$p \seqc q$ first behaves as $p$, but when $p$ terminates successfully
it continues by behaving as $q$;
\item
$\zeta \gc p$ behaves as $p$ under condition $\zeta$;
\item
$p \parc q$ behaves as the process that proceeds with $p$ and $q$ in
parallel;
\item
$p \leftm q$ behaves the same as $p \parc q$, except that it starts
with performing an action of $p$;
\item
$p \commm q$ behaves the same as $p \parc q$, except that it starts
with performing an action of $p$ and an action of $q$ synchronously;
\item
$\encap{H}(p)$ behaves the same as $p$, except that actions from $H$ are
blocked.
\end{iteml}
Intuitively, the constants and operators to build terms of sort $\Cond$
can be explained as follows:
\begin{itemize}
\item
$\eta$ is an atomic condition;
\item
$\bot$ is a condition that never holds;
\item
$\top$ is a condition that always holds;
\item
$\bcompl \zeta$ is the opposite of $\zeta$;
\item
$\zeta \join \xi$ is either $\zeta$ or $\xi$;
\item
$\zeta \meet \xi$ is both $\zeta$ and $\xi$.
\end{itemize}

Some earlier extensions of \ACP\ include conditional expressions of
the form $\cond{p}{\zeta}{q}$; see e.g.~\cite{BB92c}.
Just as in~\cite{BM05a}, we treat conditional expressions of the form
$\cond{p}{\zeta}{q}$, where $p$ and $q$ are terms of sort $\Proc$ and
$\zeta$ is a term of sort $\Cond$, as abbreviations.
That is, we write $\cond{p}{\zeta}{q}$ for
$\zeta \gc p \altc \bcompl{\zeta} \gc q$.

The axioms of \ACPec\ are given in Table~\ref{axioms-ACPec}.%
\begin{table}[!t]
\lcaption{Axioms of \ACPec}{$a,b,c \in \Actd$}
\label{axioms-ACPec}
\begin{eqntbl}
\begin{axcol}
x \altc y = y \altc x                                   & \axiom{A1}\\
(x \altc y) \altc z = x \altc (y \altc z)               & \axiom{A2}\\
x \altc x = x                                           & \axiom{A3}\\
(x \altc y) \seqc z = x \seqc z \altc y \seqc z         & \axiom{A4}\\
(x \seqc y) \seqc z = x \seqc (y \seqc z)               & \axiom{A5}\\
x \altc \dead = x                                       & \axiom{A6}\\
\dead \seqc x = \dead                                   & \axiom{A7}\\
x \seqc \ep = x                                         & \axiom{A8}\\
\ep \seqc x = x                                         & \axiom{A9}\\
{}                                                                  \\
x \parc y = x \leftm y \altc y \leftm x \altc x \commm y \altc {}
   \;\; \\ \hfill \encap{\Act}(x) \seqc \encap{\Act}(y) & \axiom{CM1T}\\
\ep \leftm x = \dead                                    & \axiom{TM2} \\
a \seqc x \leftm y = a \seqc (x \parc y)                & \axiom{CM3} \\
(x \altc y) \leftm z = x \leftm z \altc y \leftm z      & \axiom{CM4} \\
\ep \commm x = \dead                                    & \axiom{TM5} \\
x \commm \ep = \dead                                    & \axiom{TM6} \\
a \seqc x \commm b \seqc y =
                         (a \commm b) \seqc (x \parc y) & \axiom{CM7} \\
(x \altc y) \commm z = x \commm z \altc y \commm z      & \axiom{CM8} \\
x \commm (y \altc z) = x \commm y \altc x \commm z      & \axiom{CM9} \\
{}                                                                \\
a \commm b = b \commm a                                 & \axiom{C1} \\
(a \commm b) \commm c = a \commm (b \commm c)           & \axiom{C2} \\
\dead \commm a = \dead                                  & \axiom{C3}
\end{axcol}
\quad\;
\begin{axcol}
\encap{H}(\ep) = \ep                                    & \axiom{D0} \\
\encap{H}(a) = a                \hfill \mif a \not\in H & \axiom{D1} \\
\encap{H}(a) = \dead                \hfill \mif a \in H & \axiom{D2} \\
\encap{H}(x \altc y) = \encap{H}(x) \altc \encap{H}(y)  & \axiom{D3} \\
\encap{H}(x \seqc y) = \encap{H}(x) \seqc \encap{H}(y)  & \axiom{D4} \\
{}                                                                   \\
\top \gc x = x                                          & \axiom{GC1} \\
\bot \gc x = \dead                                      & \axiom{GC2} \\
\phi \gc \dead = \dead                                  & \axiom{GC3} \\
\phi \gc (x \altc y) = \phi \gc x \altc \phi \gc y      & \axiom{GC4} \\
\phi \gc x \seqc y = (\phi \gc x) \seqc y               & \axiom{GC5} \\
\phi \gc (\psi \gc x) = (\phi \meet \psi) \gc x
                                                        & \axiom{GC6} \\
(\phi \join \psi) \gc x = \phi \gc x \altc \psi \gc x
                                                        & \axiom{GC7} \\
(\phi \gc x) \leftm y = \phi \gc (x \leftm y)           & \axiom{GC8} \\
(\phi \gc x) \commm y = \phi \gc (x \commm y)           & \axiom{GC9} \\
x \commm (\phi \gc y) = \phi \gc (x \commm y)           & \axiom{GC10} \\
\encap{H}(\phi \gc x) = \phi \gc \encap{H}(x)           & \axiom{GC11} \\
{}                                                               \\
\phi \join \bot = \phi                                  & \axiom{BA1} \\
\phi \join \bcompl{\phi} = \top                         & \axiom{BA2} \\
\phi \join \psi = \psi \join \phi                       & \axiom{BA3} \\
\phi \join (\psi \meet \chi) =
              (\phi \join \psi) \meet (\phi \join \chi) & \axiom{BA4} \\
\phi \meet \top = \phi                                  & \axiom{BA5} \\
\phi \meet \bcompl{\phi} = \bot                         & \axiom{BA6} \\
\phi \meet \psi = \psi \meet \phi                       & \axiom{BA7} \\
\phi \meet (\psi \join \chi) =
              (\phi \meet \psi) \join (\phi \meet \chi) & \axiom{BA8}
\end{axcol}
\end{eqntbl}
\end{table}
CM3, CM7, C1--C3 and D1--D2 are actually axiom schemas in which $a$, $b$
and $c$ stand for arbitrary constants of \ACPec\ that differ from $\ep$
(i.e.\ $a,b,c \in \Actd$).
In D0--D4, $H$ stands for an arbitrary subset of $\Act$.
So, D0, D3 and D4 are axiom schemas as well.
Axioms A1--A9, CM1T, TM2, CM3, CM4, TM5, TM6, CM7--CM9, C1--C3 and
D0--D4 are the axioms of \ACPe.
Axioms BA1--BA8 are the axioms of Boolean Algebras (BA).
So \ACPec\ imports the (equational) axioms of both \ACPe\ and BA.
The axioms of BA have been taken from~\cite{Hal63a}.
Several alternatives for this axiomatization can be found in the
literature.
Axioms GC1--GC11 have been taken from~\cite{BB92c}, but the axiom
$\cond{x \seqc z}{\phi}{y \seqc z} = (\cond{x}{\phi}{y}) \seqc z$ (CO5)
is replaced by the simpler axiom
$\phi \gc x \seqc y = (\phi \gc x) \seqc y$ (GC5)
and similarly for axioms GC8--GC11.

The terms of sort $\Cond$ are interpreted in $\FinCondAlg$ as usual.

We proceed to the presentation of the structural operational semantics
of \ACPec.
The following relations on closed terms of sort $\Proc$ from the
language of \ACPec\ are used:
\begin{iteml}
\item
for each $\alpha \in \FinCondAlg \diff \set{\bot}$,
a unary relation ${\sterm{\alpha}}$;
\item
for each $\ell \in (\FinCondAlg \diff \set{\bot}) \x \Act$,
a binary relation ${\step{\ell}}$.
\end{iteml}
We write $\isterm{p}{\alpha}$ instead of
$p \in {\sterm{\alpha}}$ and
$\astep{p}{\gact{\alpha}{a}}{q}$ instead of
$\tup{p,q} \in {\step{\tup{\alpha,a}}}$.
The relations ${\sterm{\alpha}}$ and ${\step{\ell}}$ can be explained as
follows:
\begin{iteml}
\item
$\isterm{p}{\alpha}$:
$p$ is capable of terminating successfully under condition $\alpha$;
\item
$\astep{p}{\gact{\alpha}{a}}{q}$:
$p$ is capable of performing action $a$ under condition $\alpha$ and
then proceeding as $q$.
\end{iteml}
The structural operational semantics of \ACPec\ is
described by the transition rules given in Table~\ref{sos-ACPec}.%
\begin{table}[!t]
\caption{Transition rules for \ACPec}
\label{sos-ACPec}
\begin{ruletbl}
\Rule
{\phantom{\isterm{\top}{\ep}}}
{\isterm{\ep}{\top}}
\qquad
\Rule
{\phantom{\astep{a}{\gact{\top}{a}}{\ep}}}
{\astep{a}{\gact{\top}{a}}{\ep}}
\\
\Rule
{\isterm{x}{\phi}}
{\isterm{x \altc y}{\phi}}
\quad
\Rule
{\isterm{y}{\phi}}
{\isterm{x \altc y}{\phi}}
\quad
\Rule
{\astep{x}{\gact{\phi}{a}}{x'}}
{\astep{x \altc y}{\gact{\phi}{a}}{x'}}
\quad
\Rule
{\astep{y}{\gact{\phi}{a}}{y'}}
{\astep{x \altc y}{\gact{\phi}{a}}{y'}}
\\
\Rule
{\isterm{x}{\phi},\; \isterm{y}{\psi}}
{\isterm{x \seqc y}{\phi \meet \psi}}
\quad
\RuleC
{\isterm{x}{\phi},\; \astep{y}{\gact{\psi}{a}}{y'}}
{\astep{x \seqc y}{\gact{\phi \meet \psi}{a}}{y'}}
{\phi \meet \psi \neq \bot}
\quad
\Rule
{\astep{x}{\gact{\phi}{a}}{x'}}
{\astep{x \seqc y}{\gact{\phi}{a}}{x' \seqc y}}
\\
\RuleC
{\isterm{x}{\phi}}
{\isterm{\psi \gc x}{\phi \meet \psi}}
{\phi \meet \psi \neq \bot}
\quad
\RuleC
{\astep{x}{\gact{\phi}{a}}{x'}}
{\astep{\psi \gc x}{\gact{\phi \meet \psi}{a}}{x'}}
{\phi \meet \psi \neq \bot}
\\
\RuleC
{\isterm{x}{\phi},\; \isterm{y}{\psi}}
{\isterm{x \parc y}{\phi \meet \psi}}
{\phi \meet \psi \neq \bot}
\quad
\Rule
{\astep{x}{\gact{\phi}{a}}{x'}}
{\astep{x \parc y}{\gact{\phi}{a}}{x' \parc y}}
\quad
\Rule
{\astep{y}{\gact{\phi}{a}}{y'}}
{\astep{x \parc y}{\gact{\phi}{a}}{x \parc y'}}
\\
\RuleC
{\astep{x}{\gact{\phi}{a}}{x'},\; \astep{y}{\gact{\psi}{b}}{y'}}
{\astep{x \parc y}{\gact{\phi \meet \psi}{c}}{x' \parc y'}}
{a \commm b = c,\; \phi \meet \psi \neq \bot}
\\
\Rule
{\astep{x}{\gact{\phi}{a}}{x'}}
{\astep{x \leftm y}{\gact{\phi}{a}}{x' \parc y}}
\quad
\RuleC
{\astep{x}{\gact{\phi}{a}}{x'},\; \astep{y}{\gact{\psi}{b}}{y'}}
{\astep{x \commm y}{\gact{\phi \meet \psi}{c}}{x' \parc y'}}
{a \commm b = c,\; \phi \meet \psi \neq \bot}
\\
\Rule
{\isterm{x}{\phi}}
{\isterm{\encap{H}(x)}{\phi}}
\quad
\RuleC
{\astep{x}{\gact{\phi}{a}}{x'}}
{\astep{\encap{H}(x)}{\gact{\phi}{a}}{\encap{H}(x')}}
{a \not\in H}
\end{ruletbl}
\end{table}

\section{Transition Systems and Splitting Bisimilarity for \ACPec}
\label{sect-ts-bisim}

In this section, we adapt the definitions of conditional transition
systems and splitting bisimilarity of conditional transition systems
from~\cite{BM05a} to the presence of a process that is only capable of
terminating successfully.
In Section~\ref{sect-full-bisim-ACPec}, we will make use of conditional
transition systems and splitting bisimilarity of conditional transition
systems as defined in this section to construct models of \ACPec.

The transitions of conditional transition systems have labels that
consist of a condition different from $\bot$ and an action.
Labels of this kind are sometimes called \emph{guarded actions}.
Henceforth, we write $\CondAlgm{\kappa}$ for
$\CondAlg{\kappa} \diff \set{\bot}$.

Let $\kappa$ be an infinite cardinal.
Then a $\kappa$-\emph{conditional transition system} $T$ consists of
the following:
\begin{iteml}
\item
a set $S$ of \emph{states};
\item
a set ${\step{\ell}} \subseteq S \x S$,
for each $\ell \in \CondAlgm{\kappa} \x \Act$;
\item
a set ${\sterm{\alpha}} \subseteq S$,
for each $\alpha \in \CondAlgm{\kappa}$;
\item
an \emph{initial state} $s^0 \in S$.
\end{iteml}
If $\tup{s,s'} \in {\step{\ell}}$ for some
$\ell \in \CondAlgm{\kappa} \x \Act$, then we say that there
is a \emph{transition} from $s$ to $s'$.
We usually write $\astep{s}{\gact{\alpha}{a}}{s'}$ instead of
$\tup{s,s'} \in {\step{\tup{\alpha,a}}}$ and
$\isterm{s}{\alpha}$ instead of
$s \in {\sterm{\alpha}}$.
Furthermore, we write ${\stepf}$ for the family of sets
$(\step{\ell})_{\ell \in \CondAlgm{\kappa} \x \Act}$ and
${\stermf}$ for the family of sets
$(\sterm{\alpha})_{\alpha \in \CondAlgm{\kappa}}$.

The relations ${\sterm{\alpha}}$ and ${\step{\ell}}$ can be explained as
follows:
\begin{iteml}
\item
$\isterm{s}{\alpha}$:
in state $s$, it is possible to terminate successfully under condition
$\alpha$;
\item
$\astep{s}{\gact{\alpha}{a}}{s'}$:
in state $s$, it is possible to perform action $a$ under condition
$\alpha$, and by doing so to make a transition to state $s'$.
\end{iteml}

A conditional transition system may have states that are not reachable
from its initial state by a sequence of transitions.
Unreachable states, and the transitions between them, are not relevant
to the behaviour represented by the transition system.

Let $T = \tup{S,{\stepf},{\stermf},s^0}$ be a $\kappa$-conditional
transition system (for an infinite cardinal $\kappa$).
Then the \emph{reachability} relation of $T$ is the
smallest relation ${\reach{}{}} \subseteq S \x S$ such that:
\begin{iteml}
\item
$\reach{s}{s}$;
\item
if $\astep{s}{\ell}{s'}$ and $\reach{s'}{s''}$, then $\reach{s}{s''}$.
\end{iteml}
We write $\rs(T)$ for $\set{s \in S \where \reach{s^0}{s}}$.
$T$ is called a \emph{connected} $\kappa$-conditional transition system
if $S = \rs(T)$.

Henceforth, we will only consider connected conditional transition
systems.
However, this often calls for extraction of the connected part of a
conditional transition system resulting from composition of connected
conditional transition systems.

Let $T = \tup{S,{\stepf},{\stermf},s^0}$ be a $\kappa$-conditional
transition system (for an infinite cardinal $\kappa$) that is not
necessarily connected.
Then the \emph{connected part} of $T$, written $\conn(T)$, is defined
as follows:
\begin{ldispl}
\conn(T) = \tup{S',{\stepf}',{\stermf}',s^0}\;,
\end{ldispl}
where
\begin{ldispl}
\begin{aeqns}
S' & = & \rs(T)\;,
\end{aeqns}
\end{ldispl}
and for every $\ell \in \CondAlgm{\kappa} \x \Act$
and $\alpha \in \CondAlgm{\kappa}$:
\begin{ldispl}
\begin{aeqns}
{\stepp{\ell}} & = & {\step{\ell}} \inter (S' \x S')\;,
\\
{\stermp{\alpha}} & = & {\sterm{\alpha}} \inter S'\;.
\end{aeqns}
\end{ldispl}

It is assumed that for each infinite cardinal $\kappa$ a fixed but
arbitrary set $\States{\kappa}$ with the following properties has been
given:
\begin{iteml}
\item
the cardinality of $\States{\kappa}$ is greater than or equal to
$\kappa$;
\item
if   $S_1,S_2 \subseteq \States{\kappa}$,
then $S_1 \dunion S_2 \subseteq \States{\kappa}$
and  $S_1 \x S_2 \subseteq \States{\kappa}$.%
\footnote
{We write $A \dunion B$ for the disjoint union of sets $A$ and $B$,
 i.e.\ $A \dunion B = (A \x \set{\emptyset}) \union
                      (B \x \set{\set{\emptyset}})$.
 We write $\inj{1}$ and $\inj{2}$ for the associated injections
 $\funct{\inj{1}}{A}{A \dunion B}$ and
 $\funct{\inj{2}}{B}{A \dunion B}$, defined by
 $\inj{1}(a) = \tup{a,\emptyset}$ and
 $\inj{2}(b) = \tup{b,\set{\emptyset}}$.}
\end{iteml}

Let $\kappa$ be an infinite cardinal.
Then $\CTSe_\kappa$ is the set of all connected $\kappa$-conditional
transition systems $T = \tup{S,{\stepf},{\stermf},s^0}$ such that
$S \subset \States{\kappa}$ and the branching degree of $T$ is less
than $\kappa$, i.e.\ for all $s \in S$, the cardinality of the set
$\set{\tup{\ell,s'} \in (\CondAlgm{\kappa} \x \Act) \x S \where
      \tup{s,s'} \in {\step{\ell}}} \union
 \set{\alpha \in \CondAlgm{\kappa} \where s \in {\sterm{\alpha}}}$
is less than $\kappa$.

The condition $S \subset \States{\kappa}$ guarantees that $\CTSe_\kappa$
is indeed a set.

A conditional transition system is said to be \emph{finitely branching}
if its branching degree is less than $\aleph_0$.
Otherwise, it is said to be \emph{infinitely branching}.

The identity of the states of a conditional transition system is not
relevant to the behaviour represented by it.
Conditional transition system that differ only with respect to the
identity of the states are isomorphic.

Let $T_1 = \tup{S_1,{\stepf}_1,{\stermf}_1,s^0_1}$
and $T_2 = \tup{S_2,{\stepf}_2,{\stermf}_2,s^0_2}$
be $\kappa$-conditional transition systems
(for an infinite cardinal $\kappa$).
Then $T_1$ and $T_2$ are \emph{isomorphic}, written $T_1 \isom T_2$, if
there exists a bijective function $\funct{b}{S_1}{S_2}$ such that:
\begin{iteml}
\item
$b(s^0_1) = s^0_2$;
\item
$\astepi{s_1}{\ell}{s_1'}$ iff $\astepii{b(s_1)}{\ell}{b(s_1')}$;
\item
$\istermi{s}{\alpha}$ iff $\istermii{b(s)}{\alpha}$.
\end{iteml}
Henceforth, we will always consider two conditional transition systems
essentially the same if they are isomorphic.
\begin{remark}
\label{remark-indep-S}
The set $\CTSe_\kappa$ is independent of $\States{\kappa}$.
By that we mean the following.
Let $\CTSe_\kappa$ and $\CTSe_\kappa{}'$ result from different choices
for $\States{\kappa}$.
Then there exists a bijection $\funct{b}{\CTSe_\kappa}{\CTSe_\kappa{}'}$
such that for all $T \in \CTSe_\kappa$, $T \isom b(T)$.
\end{remark}

Bisimilarity has to be adapted to the setting with guarded actions.
In the definition given below, we use two well-known notions from the
field of Boolean algebras: a partial order relation $\beloweq$ on
$\CondAlg{\kappa}$ and a unary operation $\infjoin$ on the set of all
subsets of $\CondAlg{\kappa}$ of cardinality less than $\kappa$
(for each infinite cardinal $\kappa$).
The relation $\beloweq$ and the operation $\infjoin$ are defined by
\begin{ldispl}
\alpha \beloweq \beta\;\; \mathrm{iff}\;\; \alpha \join \beta = \beta
\qquad \mathrm{and} \qquad
\infjoin C \mbox{ is the supremum of } C \;\mathrm{in}\;
\tup{\CondAlg{\kappa},\beloweq}\;,
\end{ldispl}
respectively.
The operation $\infjoin$ is defined for all subsets of
$\CondAlg{\kappa}$ of cardinality less than $\kappa$ because
$\CondAlg{\kappa}$ is $\kappa$\,-complete.

Let $T_1 = \tup{S_1,{\stepf}_1,{\stermf}_1,s^0_1} \in \CTSe_\kappa$
and $T_2 = \tup{S_2,{\stepf}_2,{\stermf}_2,s^0_2} \in \CTSe_\kappa$
(for an infinite cardinal $\kappa$).
Then a \emph{splitting bisimulation} $B$ between $T_1$ and $T_2$ is a
binary relation $B \subseteq S_1 \x S_2$ such that $B(s^0_1,s^0_2)$ and
for all $s_1,s_2$ such that $B(s_1,s_2)$:
\begin{iteml}
\item
if $\astepi{s_1}{\gact{\alpha}{a}}{s_1'}$, then there is a set
$CS_2' \subseteq \CondAlgm{\kappa} \x S_2$ of cardinality less than
$\kappa$ such that
$\alpha \beloweq \infjoin \dom(CS_2')$ and
for all $\tup{\alpha',s_2'} \in CS_2'$,
$\astepii{s_2}{\gact{\alpha'}{a}}{s_2'}$ and $B(s_1',s_2')$;
\item
if $\astepii{s_2}{\gact{\alpha}{a}}{s_2'}$, then there is a set
$CS_1' \subseteq \CondAlgm{\kappa} \x S_1$ of cardinality less than
$\kappa$ such that
$\alpha \beloweq \infjoin \dom(CS_1')$ and
for all $\tup{\alpha',s_1'} \in CS_1'$,
$\astepi{s_1}{\gact{\alpha'}{a}}{s_1'}$ and $B(s_1',s_2')$;
\item
if $\istermi{s_1}{\alpha}$, then there is a set
$C' \subseteq \CondAlgm{\kappa}$ of cardinality less than $\kappa$ such
that $\alpha \beloweq \infjoin C'$ and for all $\alpha' \in C'$,
$\istermii{s_2}{\alpha'}$;
\item
if $\istermii{s_2}{\alpha}$, then there is a set
$C' \subseteq \CondAlgm{\kappa}$ of cardinality less than $\kappa$ such
that $\alpha \beloweq \infjoin C'$ and for all $\alpha' \in C'$,
$\istermi{s_1}{\alpha'}$.
\end{iteml}
Two conditional transition systems $T_1,T_2 \in \CTSe_\kappa$ are
\emph{splitting bisimilar}, written $T_1 \sbisim T_2$, if there exists a
splitting bisimulation $B$ between $T_1$ and $T_2$.
Let $B$ be a splitting bisimulation between $T_1$ and $T_2$.
Then we say that $B$ is a splitting bisimulation \emph{witnessing}
$T_1 \sbisim T_2$.

The name splitting bisimulation is used because a transition of one of
the related transition systems may be simulated by a set of transitions
of the other transition system.

It is easy to see that ${\sbisim}$ is an equivalence on $\CTSe_\kappa$.
Let $T \in \CTSe_\kappa$.
Then we write $\sbeqvc{T}$ for
$\set{T' \in \CTSe_\kappa \where T \sbisim T'}$, i.e.\ the
${\sbisim}$\,-equivalence class of $T$.
We write $\CTSe_\kappa / {\sbisim}$ for the set of equivalence classes
$\set{\sbeqvc{T} \where T \in \CTSe_\kappa}$.

In Section~\ref{sect-full-bisim-ACPec}, we will use $\CTSe_\kappa$ as
the domain of a structure that is a model of \ACPec.
As the domain of a structure, $\CTSe_\kappa / {\sbisim}$ must be a set.
That is the case because $\CTSe_\kappa$ is a set.
The latter is guaranteed by considering only conditional transition
systems of which the set of states is a subset of $\States{\kappa}$.
\begin{remark}
\label{remark-largeness}
The question arises whether $\States{\kappa}$ is large enough if its
cardinality is greater than or equal to $\kappa$.
This question can be answered in the affirmative.
Let $T = \tup{S,{\stepf},{\stermf},s^0}$ be a connected
$\kappa$-conditional transition system of which the branching degree is
less than $\kappa$.
Then there exists a connected $\kappa$-conditional transition system
$T' = \tup{S',{\stepf}',{\stermf}',s^0{}'}$ of which the branching
degree is less than $\kappa$ such that $T \sbisim T'$ and the
cardinality of $S'$ is less than $\kappa$.
\end{remark}

It is easy to see that, if we would consider conditional transition
systems with unreachable states as well, each conditional transition
system would be splitting bisimilar to its connected part.
It is also easy to see that isomorphic conditional transition systems
are splitting bisimilar.

\section{Full Splitting Bisimulation Models of \ACPec}
\label{sect-full-bisim-ACPec}

In this section, we introduce the full splitting bisimulation models of
\ACPec.
They are models of which the domain consists of equivalence classes of
conditional transition systems modulo splitting bisimilarity.
The qualification ``full'' expresses that there exist other splitting
bisimulation models, but each of them is isomorphically embedded in a
full splitting bisimulation model.

The models of \ACPec\ are structures that consist of the following:
\begin{iteml}
\item
a non-empty set $\cD$, called the \emph{domain} of the model;
\item
for each constant of \ACPec, an element of $\cD$;
\item
for each $n$-ary operator of \ACPec, an $n$-ary operation on $\cD$.
\end{iteml}
In the full splitting bisimulation models of \ACPec\ that are introduced
in this section, the domain is $\CTSe_\kappa / {\sbisim}$ for some
infinite cardinal $\kappa$.
We obtain the models concerned by associating certain elements of
$\CTSe_\kappa / {\sbisim}$ with the constants of \ACPec\ and certain
operations on $\CTSe_\kappa / {\sbisim}$ with the operators of \ACPec.
We begin by associating elements of $\CTSe_\kappa$ and operations on
$\CTSe_\kappa$ with the constants and operators.
The result of this is subsequently lifted to $\CTSe_\kappa / {\sbisim}$.

It is assumed that for each infinite cardinal $\kappa$ a fixed but
arbitrary function
$\funct{\cf{\kappa}}
       {(\setof{(\States{\kappa})} \diff \emptyset)}{\States{\kappa}}$
such that for all $S \in \setof{(\States{\kappa})} \diff \emptyset$,
$\cf{\kappa}(S) \in S$ has been given.

We associate with each constant $c$ of \ACPec\ an element
$\what{c}$ of $\CTSe_\kappa$ and with each operator $f$ of \ACPec\
an operation $\what{f}$ on $\CTSe_\kappa$ as follows.
\begin{iteml}
\item
\begin{ldispl}
\tsdead = \tup{\set{s^0},\emptyset,\emptyset,s^0}\;,
\end{ldispl}
where
\begin{ldispl}
\begin{aeqns}
s^0 & = & \cf{\kappa}(\States{\kappa})\;.
\end{aeqns}
\end{ldispl}
\item
\begin{ldispl}
\tsep = \tup{\set{s^0},\emptyset,{\stermf},s^0}\;,
\end{ldispl}
where
\begin{ldispl}
\begin{aeqns}
s^0 & = & \cf{\kappa}(\States{\kappa})\;,
\\
{\sterm{\top}} & = & \set{s^0}\;,
\end{aeqns}
\end{ldispl}
and for every $\alpha \in \CondAlgm{\kappa} \diff \set{\top}$:
\begin{ldispl}
{\sterm{\alpha}} = \emptyset\;.
\end{ldispl}
\item
\begin{ldispl}
\tsact{a} = \tup{\set{s^0,\stterm},{\stepf},{\stermf},s^0}\;,
\end{ldispl}
where
\begin{ldispl}
\begin{aeqns}
s^0 & = & \cf{\kappa}(\States{\kappa})\;,
\\
\stterm & = & \cf{\kappa}(\States{\kappa} \diff \set{s^0})\;,
\\
{\step{\gact{\top}{a}}} & = & \set{\tup{s^0,\stterm}}\;,
\\
{\sterm{\top}} & = & \set{\stterm}\;,
\end{aeqns}
\end{ldispl}
and for every
$\tup{\alpha',a'} \in
 (\CondAlgm{\kappa} \x \Act) \diff \set{\tup{\top,a}}$ and
$\alpha'' \in \CondAlgm{\kappa} \diff \set{\top}$:
\begin{ldispl}
\begin{aeqns}
{\step{\gact{\alpha'}{a'}}} & = &  \emptyset\;,
\\
{\sterm{\alpha''}} & = & \emptyset\;.
\end{aeqns}
\end{ldispl}
\item
Let $T_i = \tup{S_i,{\stepf}_i,{\stermf}_i,s^0_i} \in \CTSe_\kappa$ for
$i = 1,2$.
Then
\begin{ldispl}
T_1 \tsaltc T_2 = \conn(S,{\stepf},{\stermf},s^0)\;,
\end{ldispl}
where
\begin{ldispl}
\begin{aeqns}
s^0 & = & \cf{\kappa}(\States{\kappa} \diff (S_1 \dunion S_2))\;,
\\
S   & = & \set{s^0} \union (S_1 \dunion S_2)\;,
\end{aeqns}
\end{ldispl}
and for every
$\tup{\alpha,a} \in \CondAlgm{\kappa} \x \Act$ and
$\alpha' \in \CondAlgm{\kappa}$:
\begin{ldispl}
\begin{aeqns}
{\step{\tup{\alpha,a}}} & = &
\set{\tup{s^0,\inj{1}(s)} \where \astepi{s^0_1}{\gact{\alpha}{a}}{s}}
\\ & \union &
\set{\tup{s^0,\inj{2}(s)} \where \astepii{s^0_2}{\gact{\alpha}{a}}{s}}
\\ & \union &
\set{\tup{\inj{1}(s),\inj{1}(s')} \where
     \astepi{s}{\gact{\alpha}{a}}{s'}}
\\ & \union &
\set{\tup{\inj{2}(s),\inj{2}(s')} \where
     \astepii{s}{\gact{\alpha}{a}}{s'}}\;,
\\
{\sterm{\alpha'}} & = &
\set{s^0 \where \istermi{s^0_1}{\alpha'}}
\\ & \union &
\set{s^0 \where \istermii{s^0_2}{\alpha'}}
\\ & \union &
\set{\inj{1}(s)  \where \istermi{s}{\alpha'}}
\\ & \union &
\set{\inj{2}(s) \where \istermii{s}{\alpha'}}\;.
\end{aeqns}
\end{ldispl}
\item
Let $T_i = \tup{S_i,{\stepf}_i,{\stermf}_i,s^0_i} \in \CTSe_\kappa$ for
$i = 1,2$.
Then
\begin{ldispl}
T_1 \tsseqc T_2 = \conn(S,{\stepf},{\stermf},s^0_1)\;,
\end{ldispl}
where
\begin{ldispl}
\begin{aeqns}
S & = & S_1 \dunion S_2\;,
\end{aeqns}
\end{ldispl}
and for every
$\tup{\alpha,a} \in \CondAlgm{\kappa} \x \Act$ and
$\alpha' \in \CondAlgm{\kappa}$:
\begin{ldispl}
\begin{aeqns}
{\step{\tup{\alpha,a}}} & = &
\set{\tup{\inj{1}(s),\inj{1}(s')} \where
     \astepi{s}{\gact{\alpha}{a}}{s'} \And
     \Not \Exists{\beta}{\istermi{s'}{\beta}}}
\\ & \union &
\set{\tup{\inj{1}(s),\inj{2}(s^0_2)} \where
     \Exists{s',\beta}{\astepi{s}{\gact{\alpha}{a}}{s'} \And
                       \istermi{s'}{\beta}}}
\\ & \union &
\set{\tup{\inj{2}(s^0_2),\inj{2}(s')} \where {}
\\ & & \hfill
     \Exists{s,\beta,\beta'}
      {\istermi{s}{\beta} \And
       \astepii{s^0_2}{\gact{\beta'}{a}}{s'} \And
       \alpha = \beta \meet \beta'}}
\\ & \union &
\set{\tup{\inj{2}(s),\inj{2}(s')} \where
     \astepii{s}{\gact{\alpha}{a}}{s'} \And s \neq s^0_2}\;,
\\
{\sterm{\alpha'}} & = &
\set{\inj{2}(s^0_2) \where
     \Exists{s,\beta,\beta'}
      {\istermi{s}{\beta} \And
       \istermii{s^0_2}{\beta'} \And
       \alpha' = \beta \meet \beta'}}
\\ & \union &
\set{\inj{2}(s) \where \istermii{s}{\alpha'} \And s \neq s^0_2}\;.
\end{aeqns}
\end{ldispl}
\item
Let $T = \tup{S,{\stepf},{\stermf},s^0} \in \CTSe_\kappa$.
Then
\begin{ldispl}
\alpha \tsgc T = \conn(S,{\stepf}',{\stermf}',s^0)\;,
\end{ldispl}
where for every
$\tup{\alpha',a} \in \CondAlgm{\kappa} \x \Act$ and
$\alpha'' \in \CondAlgm{\kappa}$:
\begin{ldispl}
\begin{aeqns}
{\step{\tup{\alpha',a}}}{}' & = &
\set{\tup{s^0,s'} \where
     \Exists{\beta}{\astep{s^0}{\gact{\beta}{a}}{s'} \And
                                \alpha' = \alpha \meet \beta}} \\
& \union &
\set{\tup{s,s'} \where
     \astep{s}{\gact{\alpha'}{a}}{s'} \And s \neq s^0}\;,
\\
{\stermp{\alpha''}} & = &
\set{s^0 \where
     \Exists{\beta}{\isterm{s^0}{\beta} \And
                               \alpha'' = \alpha \meet \beta}} \\
& \union &
\set{s \where \isterm{s}{\alpha''} \And s \neq s^0}\;.
\end{aeqns}
\end{ldispl}
\item
Let $T_i = \tup{S_i,{\stepf}_i,{\stermf}_i,s^0_i} \in \CTSe_\kappa$ for
$i = 1,2$.
Then
\begin{ldispl}
T_1 \tsparc T_2 = \tup{S,{\stepf},{\stermf},s^0}\;,
\end{ldispl}
where
\begin{ldispl}
\begin{aeqns}
s^0 & = & \tup{s^0_1,s^0_2}\;,
\eqnsep
S & = & S_1  \x S_2\;,
\end{aeqns}
\end{ldispl}
and for every
$\tup{\alpha,a} \in \CondAlgm{\kappa} \x \Act$ and
$\alpha'' \in \CondAlgm{\kappa}$:
\begin{ldispl}
\begin{aeqns}
{\step{\tup{\alpha,a}}} & = &
\set{\tup{\tup{s_1,s_2},\tup{s_1',s_2}} \where
     \astepi{s_1}{\gact{\alpha}{a}}{s_1'} \And
     s_2 \in S_2}
\\ & \union &
\set{\tup{\tup{s_1,s_2},\tup{s_1,s_2'}} \where
     s_1 \in S_1 \And
     \astepii{s_2}{\gact{\alpha}{a}}{s_2'}}
\\ & \union &
\set{\tup{\tup{s_1,s_2},\tup{s_1',s_2'}} \where
\\ & & \phantom{\{\,}
\smash{\OR{\alpha',\beta' \in \CondAlgm{\kappa},\, a',b' \in \Act}}
        (\astepi{s_1}{\gact{\alpha'}{a'}}{s_1'} \And
         \astepii{s_2}{\gact{\beta'}{b'}}{s_2'} \And\;\;
\\ & & \hfill
         \alpha' \meet \beta' = \alpha \And a' \commm b' = a)}\;,
\eqnsep
{\sterm{\alpha''}} & = &
\set{\tup{s_1,s_2} \where
     \OR{\alpha',\beta' \in \CondAlgm{\kappa}}
       (\istermi{s_1}{\alpha'} \And \istermii{s_2}{\beta'} \And
        \alpha' \meet \beta' = \alpha'')}\;.
\end{aeqns}
\end{ldispl}
\item
Let $T_i = \tup{S_i,{\stepf}_i,{\stermf}_i,s^0_i} \in \CTSe_\kappa$ for
$i = 1,2$.
Suppose that $T_1 \tsparc T_2 = \tup{S,{\stepf},{\stermf},s^0}$.
Then
\begin{ldispl}
T_1 \tsleftm T_2 = \conn(S',{\stepf}{}',{\stermf},s^0{}')\;,
\end{ldispl}
where
\begin{ldispl}
\begin{aeqns}
s^0{}' & = & \cf{\kappa}(\States{\kappa} \diff S)\;,
\eqnsep
S' & = & \set{s^0{}'} \union S\;,
\end{aeqns}
\end{ldispl}
and for every
$\tup{\alpha,a} \in \CondAlgm{\kappa} \x \Act$:
\begin{ldispl}
\begin{aeqns}
{\step{\tup{\alpha,a}}}{}' & = &
\set{\tup{s^0{}',\tup{s,s^0_2}} \where
     \astepi{s^0_1}{\gact{\alpha}{a}}{s}} \union
{\step{\tup{\alpha,a}}}\;.
\end{aeqns}
\end{ldispl}
\item
Let $T_i = \tup{S_i,{\stepf}_i,{\stermf}_i,s^0_i} \in \CTSe_\kappa$ for
$i = 1,2$.
Suppose that $T_1 \tsparc T_2 = \tup{S,{\stepf},{\stermf},s^0}$.
Then
\begin{ldispl}
T_1 \tscommm T_2 = \conn(S',{\stepf}',{\stermf},s^0{}')\;,
\end{ldispl}
where
\begin{ldispl}
\begin{aeqns}
s^0{}' & = & \cf{\kappa}(\States{\kappa} \diff S)\;,
\eqnsep
S' & = & \set{s^0{}'} \union S\;,
\end{aeqns}
\end{ldispl}
and for every
$\tup{\alpha,a} \in \CondAlgm{\kappa} \x \Act$:
\begin{ldispl}
\begin{aeqns}
{\step{\tup{\alpha,a}}}{}' & = &
\set{\tup{s^0{}',\tup{s_1,s_2}} \where
\\ & & \phantom{\{\,}
\smash{\OR{\alpha',\beta' \in \CondAlgm{\kappa},\, a',b' \in \Act}}
        (\astepi{s^0_1}{\gact{\alpha'}{a'}}{s_1} \And
         \astepii{s^0_2}{\gact{\beta'}{b'}}{s_2} \And {}\phantom{)\;.}
\\ & & \hfill
         \alpha' \meet \beta' = \alpha \And a' \commm b' = a)}\;.
\end{aeqns}
\end{ldispl}
\item
Let $T = \tup{S,{\stepf},{\stermf},s^0} \in \CTSe_\kappa$.
Then
\begin{ldispl}
\tsencap{H}(T) = \conn(S,{\stepf}',{\stermf},s^0)\;,
\end{ldispl}
where for every
$\tup{\alpha,a} \in \CondAlgm{\kappa} \x (\Act \diff H)$:
\begin{ldispl}
\begin{aeqns}
{\step{\tup{\alpha,a}}}{}' & = & {\step{\tup{\alpha,a}}}\;,
\end{aeqns}
\end{ldispl}
and for every
$\tup{\alpha,a} \in \CondAlgm{\kappa} \x H$:
\begin{ldispl}
\begin{aeqns}
{\step{\tup{\alpha,a}}}{}' & = & \emptyset\;.
\end{aeqns}
\end{ldispl}
\end{iteml}
In the definition of alternative composition on $\CTSe_\kappa$, a new
initial state is introduced because, in $T_1$ and/or $T_2$, there may
exist a transition back to the initial state.
The connected part of the resulting conditional transition system is
extracted because the initial states of $T_1$ and $T_2$ may be
unreachable from the new initial state.

\begin{remark}
\label{remark-indep-Ch}
The elements of $\CTSe_\kappa$ and the operations on $\CTSe_\kappa$
defined above are independent of $\cf{\kappa}$.
Different choices for $\cf{\kappa}$ lead for each constant of \ACPec\
to isomorphic elements of $\CTSe_\kappa$ and lead for each operator
\ACPec\ to operations on $\CTSe_\kappa$ with isomorphic results.
\end{remark}

We can show that splitting bisimilarity is a congruence with respect to
the operations on $\CTSe_\kappa$ associated with the operators of
\ACPec.
\begin{proposition}[Congruence]
\label{prop-congruence-ACPec}
Let $\kappa$ be an infinite cardinal.
Then for all $T_1,T_2,T_1',T_2' \in \CTSe_\kappa$ and
$\alpha \in \CondAlg{\kappa}$,
$T_1 \sbisim T_1'$ and $T_2 \sbisim T_2'$ imply
$T_1 \tsaltc T_2 \sbisim T_1' \tsaltc T_2'$,
$T_1 \tsseqc T_2 \sbisim T_1' \tsseqc T_2'$,
$\alpha \tsgc T_1 \sbisim \alpha \tsgc T_1'$,
$T_1 \tsparc T_2 \sbisim T_1' \tsparc T_2'$,
$T_1 \tsleftm T_2 \sbisim T_1' \tsleftm T_2'$,
$T_1 \tscommm T_2 \sbisim T_1' \tscommm T_2'$ and
$\tsencap{H}(T_1) \sbisim \tsencap{H}(T_1')$.
\end{proposition}
\begin{proof}
For all operations except $\tsparc$, witnessing splitting bisimulations
are constructed in the same way as in the congruence proofs for the
corresponding operations on $\CTS_\kappa$ given in~\cite{BM05a}.
For $\tsparc$, the construction of a witnessing splitting bisimulation
is easier than in~\cite{BM05a}.%
\footnote
{Because the relation constructed in~\cite{BM05a} is by mistake the same
 as the one constructed in this paper, we should actually say ``in the
 revision of~\cite{BM05a} that can be found at
 \texttt{www.win.tue.nl/\~{}keesm/sbrc.pdf}''.
 }
Let $R_1$ and $R_2$ be splitting bisimulations witnessing
$T_1 \sbisim T_1'$ and $T_2 \sbisim T_2'$, respectively.
Then we construct relations $R_{\tsparc}$ as follows:
\begin{iteml}
\item
$R_{\tsparc} =
 \set{\tup{\tup{s_1,s_2},\tup{s_1',s_2'}} \where
      \tup{s_1,s_1'} \in R_1, \tup{s_2,s_2'} \in R_2}$.
\end{iteml}
Given the definition of parallel composition, it is easy to see that
$R_{\tsparc}$ is a splitting bisimulation witnessing
$T_1 \tsparc T_2 \sbisim T_1' \tsparc T_2'$.
\qed
\end{proof}

The \emph{full splitting bisimulation models} $\CPrce_\kappa$, one for
each infinite cardinal $\kappa$, consist of the following:
\begin{iteml}
\item
a set $\cP$, called the domain of $\CPrce_\kappa$;
\item
for each constant $c$ of \ACPec, an element $\wtilde{c}$ of $\cP$;
\item
for each $n$-ary operator $f$ of \ACPec, an $n$-ary operation
$\wtilde{f}$ on $\cP$;
\end{iteml}
where those ingredients are defined as follows:
\begin{ldispl}
\begin{aeqns}
\cP & = & \CTSe_\kappa / {\sbisim}\;,
\eqnsep
\sdead & = & \sbeqvc{\tsdead}\;,
\eqnsep
\sep & = & \sbeqvc{\tsep}\;,
\eqnsep
\sact{a} & = & \sbeqvc{\tsact{a}}\;,
\eqnsep
\sbeqvc{T_1} \saltc \sbeqvc{T_2} & = & \sbeqvc{T_1 \tsaltc T_2}\;,
\eqnsep
\sbeqvc{T_1} \sseqc \sbeqvc{T_2} & = & \sbeqvc{T_1 \tsseqc T_2}\;,
\end{aeqns}
\qquad\qquad
\begin{aeqns}
{}
\eqnsep
\alpha \sgc \sbeqvc{T_1} & = & \sbeqvc{\alpha \tsgc T_1}\;.
\eqnsep
\sbeqvc{T_1} \sparc \sbeqvc{T_2} & = & \sbeqvc{T_1 \tsparc T_2}\;,
\eqnsep
\sbeqvc{T_1} \sleftm \sbeqvc{T_2} & = & \sbeqvc{T_1 \tsleftm T_2}\;,
\eqnsep
\sbeqvc{T_1} \scommm \sbeqvc{T_2} & = & \sbeqvc{T_1 \tscommm T_2}\;,
\eqnsep
\sencap{H}(\sbeqvc{T_1}) & = & \sbeqvc{\tsencap{H}(T_1)}\;.
\end{aeqns}
\end{ldispl}
The operations on $\CTSe_\kappa / {\sbisim}$ are well-defined because
$\sbisim$ is a congruence with respect to the corresponding operations
on $\CTSe_\kappa$.

The structures $\CPrce_\kappa$ are models of \ACPec.
\begin{theorem}[Soundness of \ACPec]
\label{theorem-soundness-ACPec}
For each infinite cardinal $\kappa$, we have
$\Sat{\CPrce_\kappa}{\ACPec}$.
\end{theorem}
\begin{proof}
Because ${\sbisim}$ is a congruence, it is sufficient to show that all
additional axioms are sound.
The soundness of all additional axioms follows easily from the
definition of $\CPrce_\kappa$.
\qed
\end{proof}
For all axioms that are in common with \ACPc, the proof of soundness
with respect to $\CPrce_\kappa$ follows the same line as the proof of
soundness with respect to $\CPrc_\kappa$.

The full splitting bisimulation models are related by isomorphic
embeddings.
\begin{theorem}[Isomorphic Embedding]
\label{theorem-embedding-ts-models}
Let $\kappa$ and $\kappa'$ be infinite cardinals such that
$\kappa < \kappa'$.
Then $\CPrce_\kappa$ is isomorphically embedded in $\CPrce_{\kappa'}$.
\end{theorem}
\begin{proof}
The proof is analogous to the proof of the corresponding property for
the full splitting bisimulation models of \ACPc\ given in~\cite{BM05a}.
\qed
\end{proof}

\section{SOS-Based Splitting Bisimilarity for \ACPec}
\label{sect-sos-based-bisim}

It is customary to associate transition systems with closed terms of
the language of an \ACP-like theory about processes by means of
structural operational semantics and to identify closed terms if their
associated transition systems are splitting bisimilar.

The structural operational semantics of \ACPec\ presented in
Section~\ref{sect-ACPec} determines a conditional transition system
for each process that can be denoted by a closed term of sort $\Proc$.
These transition systems are special in the sense that their states are
closed terms of sort $\Proc$.

Let $p$ be a closed term of sort $\Proc$.
Then the transition system of $p$ \emph{induced by} the structural
operational semantics of \ACPec, written $\sym{CTS}(p)$, is the
connected conditional transition system
$\conn(S,{\stepf},{\stermf},s^0)$, where:
\begin{iteml}
\item
$S$ is the set of all closed terms of sort $\Proc$;
\item
the sets ${\step{\tup{\alpha,a}}} \subseteq S \x S$ and
${\sterm{\alpha}} \subseteq S$ for each
$\alpha \in \FinCondAlg \diff \set{\bot}$ and $a \in \Act$ are the
smallest subsets of $S \x S$ and $S$, respectively, for which the
transition rules from Table~\ref{sos-ACPec} hold;
\item
$s^0 \in S$ is the closed term $p$.
\end{iteml}

Let $p$ and $q$ be closed terms of sort $\Proc$.
Then we say that $p$ and $q$ are \emph{splitting bisimilar}, written
$p \sbisim q$, if $\sym{CTS}(p) \sbisim \sym{CTS}(q)$.

Clearly, the structural operational semantics does not give rise to
infinitely branching conditional transition systems.
For each closed term $p$ of sort $\Proc$, there exists a
$T \in \CTSe_{\aleph_0}$ such that $\sym{CTS}(p) \isom T$.
In Section~\ref{sect-full-bisim-ACPec}, it has been shown that it is
possible to consider infinitely branching conditional transition
systems as well.

\section{Guarded Recursion}
\label{sect-recursion}

In order to allow for the description of (potentially) non-terminating
processes, we add guarded recursion to \ACPec.

A \emph{recursive specification} over \ACPec\ is a set of equations
$E = \set{X = t_X \where X \in V}$ where $V$ is a set of variables and
each $t_X$ is a term of sort $\Proc$ that only contains variables from
$V$.
We write $\vars(E)$ for the set of all variables that occur on the
left-hand side of an equation in $E$.
A \emph{solution} of a recursive specification $E$ is a set of processes
(in some model of \ACPec) $\set{P_X \where X \in \vars(E)}$ such that
the equations of $E$ hold if, for all $X \in \vars(E)$, $X$ stands for
$P_X$.

Let $t$ be a term of sort $\Proc$ containing a variable $X$.
We call an occurrence of $X$ in $t$ \emph{guarded} if $t$ has a subterm
of the form $a \seqc t'$ containing this occurrence of $X$.
A recursive specification over \ACPec\ is called a
\emph{guarded} recursive specification if all occurrences of variables
in the right-hand sides of its equations are guarded or it can be
rewritten to such a recursive specification using the axioms of
\ACPec\ and the equations of the recursive specification.
We are only interested in models of \ACPec\ in which guarded recursive
specifications have unique solutions.

For each guarded recursive specification $E$ and each variable
$X \in \vars(E)$, we introduce a constant of sort $\Proc$ standing for
the unique solution of $E$ for $X$.
This constant is denoted by $\rec{X}{E}$.
We often write $X$ for $\rec{X}{E}$ if $E$ is clear from the context.
In such cases, it should also be clear from the context that we use $X$
as a constant.

We will also use the following notation.
Let $t$ be a term of sort $\Proc$ and $E$ be a guarded recursive
specification over \ACPec.
Then we write $\rec{t}{E}$ for $t$ with, for all $X \in \vars(E)$, all
occurrences of $X$ in $t$ replaced by $\rec{X}{E}$.

The additional axioms for recursion are the equations given
in Table~\ref{axioms-ACPec-rec}.%
\begin{table}[!t]
\caption{Axioms for recursion}
\label{axioms-ACPec-rec}
\begin{eqntbl}
\begin{saxcol}
\rec{X}{E} = \rec{t_X}{E} & \mif X = t_X \in E
& \axiom{RDP}
\\
E \Then X = \rec{X}{E}    & \mif X \in \vars(E)
& \axiom{RSP}
\end{saxcol}
\end{eqntbl}
\end{table}
Both RDP and RSP are axiom schemas.
A side condition is added to restrict the variables, terms and guarded
recursive specifications for which $X$, $t_X$ and $E$ stand.
The additional axioms for recursion are known as the recursive
definition principle (RDP) and the recursive specification principle
(RSP).
The equations $\rec{X}{E} = \rec{t_X}{E}$ for a fixed $E$ express that
the constants $\rec{X}{E}$ make up a solution of $E$.
The conditional equations $E \Then X = \rec{X}{E}$ express that this
solution is the only one.

The structural operational semantics for the constants $\rec{X}{E}$ is
described by the transition rules given in Table~\ref{sos-ACPec-rec}.
\begin{table}[!t]
\caption{Transition rules for recursion}
\label{sos-ACPec-rec}
\begin{ruletbl}
\RuleC
{\isterm{\rec{t_X}{E}}{\phi}}
{\isterm{\rec{X}{E}}{\phi}}
{X \!=\! t_X \,\in\, E}
\qquad
\RuleC
{\astep{\rec{t_X}{E}}{\gact{\phi}{a}}{x'}}
{\astep{\rec{X}{E}}{\gact{\phi}{a}}{x'}}
{X \!=\! t_X \,\in\, E}
\end{ruletbl}
\end{table}

In the full splitting bisimulation models of \ACPec, guarded recursive
specifications over \ACPec\ have unique solutions.
\begin{theorem}[Unique solutions in $\CPrce_\kappa$]
\label{theorem-uniqueness}
\sloppy
For each infinite cardinal $\kappa$, guarded recursive specifications
over \ACPec\ have unique solutions in $\CPrce_\kappa$.
\end{theorem}
\begin{proof}
The proof is analogous to the proof of the corresponding property for
the full splitting bisimulation models of \ACPc\ given in~\cite{BM05a}.
\qed
\end{proof}
Thus, the full splitting bisimulation models $\CPrce_\kappa{}'$ of
\ACPec\ with guarded recursion are simply the expansions of the full
splitting bisimulation models $\CPrce_\kappa$ of \ACPec\ obtained by
associating with each constant $\rec{X}{E}$ the unique solution of $E$
for $X$ in the full splitting bisimulation model concerned.

\section{Evaluation of Conditions}
\label{sect-cond-eval}

Guarded commands cannot always be eliminated from closed terms of
sort $\Proc$ because conditions different from both $\bot$ and $\top$
may be involved.
The condition evaluation operators introduced below, can be brought
into action in such cases.
These operators require to fix an infinite cardinal $\lambda$.
By doing so, full splitting bisimulation models with domain
$\CTSe_{\kappa} / {\sbisim}$ for $\kappa > \lambda$ are excluded.

There are unary $\lambda$-\emph{complete condition evaluation} operators
$\funct{\ceval{h}}{\Proc}{\Proc}$ and $\funct{\ceval{h}}{\Cond}{\Cond}$
for each $\lambda$-complete endomorphisms $h$ of $\CondAlg{\lambda}$.%
\footnote
{For a definition of $\kappa$-complete endomorphisms, see
 e.g.~\cite{MB89a}.}

These operators can be explained as follows:
$\ceval{h}(p)$ behaves as $p$ with each condition $\zeta$ occurring in
$p$ replaced according to $h$.
If the image of $\CondAlg{\lambda}$ under $h$ is $\Bool$, i.e.\ the
Boolean algebra with domain $\set{\bot,\top}$, then guarded commands
can be eliminated from $\ceval{h}(p)$.
In the case where the image of $\CondAlg{\lambda}$ under $h$ is not
$\Bool$, $\ceval{h}$ can be regarded to evaluate the conditions only
partially.

Henceforth, we write $\CondHom{\lambda}$ for the set of all
$\lambda$-complete endomorphisms of $\CondAlg{\lambda}$.

The additional axioms for $\ceval{h}$, where $h \in \CondHom{\lambda}$,
are the axioms given in Table~\ref{axioms-ceval}.
\begin{table}[!t]
\lcaption{Axioms for condition evaluation}
  {$a \in \Actd$, $\eta \in \ACond$,
   $\eta' \in \ACond \union \set{\bot,\top}$}
\label{axioms-ceval}
\begin{eqntbl}
\begin{axcol}
\ceval{h}(\ep) = \ep                                   & \axiom{CE1T}\\
\ceval{h}(a \seqc x) = a \seqc \ceval{h}(x)            & \axiom{CE2} \\
\ceval{h}(x \altc y) = \ceval{h}(x) \altc \ceval{h}(y) & \axiom{CE3} \\
\ceval{h}(\phi \gc x) =
                      \ceval{h}(\phi) \gc \ceval{h}(x) & \axiom{CE4} \\
\ceval{h}(\ceval{h'}(x)) = \ceval{h \comp h'}(x)       & \axiom{CE5}
\end{axcol}
\quad\quad
\begin{axcol}
\ceval{h}(\bot) = \bot                                 & \axiom{CE6} \\
\ceval{h}(\top) = \top                                 & \axiom{CE7} \\
\ceval{h}(\eta) = \eta'    \hfill \mif h(\eta) = \eta' & \axiom{CE8} \\
\ceval{h}(\bcompl \phi) = \bcompl \ceval{h}(\phi)      & \axiom{CE9} \\
\ceval{h}(\phi \join \psi) =
                 \ceval{h}(\phi) \join \ceval{h}(\psi) & \axiom{CE10} \\
\ceval{h}(\phi \meet \psi) =
                 \ceval{h}(\phi) \meet \ceval{h}(\psi) & \axiom{CE11}
\end{axcol}
\end{eqntbl}
\end{table}

The structural operational semantics of \ACPec\ extended with
condition evaluation is described by the transition rules for \ACPec\
and the transition rules given in Table~\ref{sos-ceval}.
\begin{table}[!t]
\caption{Transition rules for condition evaluation}
\label{sos-ceval}
\begin{ruletbl}
\RuleC
{\isterm{x}{\phi}}
{\isterm{\ceval{h}(x)}{h(\phi)}}
{h(\phi) \neq \bot}
\qquad
\RuleC
{\astep{x}{\gact{\phi}{a}}{x'}}
{\astep{\ceval{h}(x)}{\gact{h(\phi)}{a}}{\ceval{h}(x')}}
{h(\phi) \neq \bot}
\end{ruletbl}
\end{table}

If $\lambda$ is a regular infinite cardinal, the elements of
$\CondAlg{\lambda}$ can be used to represent equivalence classes with
respect to logical equivalence of the set of all propositions with
elements of $\ACond$ as propositional variables and with conjunctions
and disjunctions of less than $\lambda$ propositions.
We write $\PropAlg{\lambda}$ for this set of propositions.
If $\lambda$ is a regular infinite cardinal, it is likely that there is
a theory $\Phi$ about the atomic conditions in the shape of a set of
propositions.
Let $\Phi \subset \PropAlg{\lambda}$, and
let $h_\Phi \in \CondHom{\lambda}$ be such that for all
$\alpha,\beta \in \CondAlg{\lambda}$:
\begin{equation}
\Der{\Phi}{\repr{h_\Phi(\alpha)} \Iff \repr{\alpha}}
\quad \mathrm{and} \quad
h_\Phi(\alpha) = h_\Phi(\beta)\;\; \mathrm{iff}\;\;
\Der{\Phi}{\repr{\alpha} \Iff \repr{\beta}}
\label{endo-theory}
\end{equation}
where $\repr{\alpha}$ is a representative of the equivalence class of
propositions isomorphic to $\alpha$.
Then we have $h_\Phi(\alpha) = \top$ iff $\repr{\alpha}$ is derivable
from $\Phi$ and $h_\Phi(\alpha) = \bot$ iff $\Not \repr{\alpha}$ is
derivable from $\Phi$.
The image of $\CondAlg{\lambda}$ under $h_\Phi$ is $\Bool$ iff $\Phi$ is
a complete theory.
If $\Phi$ is not a complete theory, then $h_\Phi$ is not uniquely
determined by~(\ref{endo-theory}).
However, the images of $\CondAlg{\lambda}$ under the different
endomorphisms satisfying~(\ref{endo-theory}) are isomorphic subalgebras
of $\CondAlg{\lambda}$.
Moreover, if both $h$ and $h'$ satisfy~(\ref{endo-theory}), then
$\Der{\Phi}{\repr{h(\alpha)} \Iff \repr{h'(\alpha)}}$ for all
$\alpha \in \CondAlg{\lambda}$.

Below, we show that condition evaluation on the basis of a complete
theory can be viewed as substitution on the basis of the theory.
That leads us to the use of the following convention:
for $\alpha \in \FinCondAlg$,\, $\cterm{\alpha}$ stands for an arbitrary
closed term of sort $\Cond$ of which the value in $\FinCondAlg$ is
$\alpha$.
\begin{proposition}[Condition evaluation on the basis of a theory]
\label{prop-eval-theory}
Assume that $\lambda$ is a regular infinite cardinal.
Let $\Phi \subset \PropAlg{\lambda}$ be a complete theory and let $p$ be
a closed term of sort $\Proc$.
Then $\ceval{h_\Phi}(p) = p'$ where $p'$ is $p$ with, for all
$\alpha \in \FinCondAlg$, in all subterms of the form
$\cterm{\alpha} \gc q$,\, $\cterm{\alpha}$ replaced by $\top$ if
$\Der{\Phi}{\repr{\alpha}}$ and $\cterm{\alpha}$
replaced by $\bot$ if $\Der{\Phi}{\Not \repr{\alpha}}$.
\end{proposition}
\begin{proof}
This result follows immediately from the definition of $h_\Phi$ and the
distributivity of $\ceval{h_\Phi}$ over all operators of \ACPec.
\qed
\end{proof}
In $\mu$CRL~\cite{GP94a}, an extension of \ACP\ which includes
conditional expressions, we find a formalization of the
substitution-based alternative for $\ceval{h_\Phi}$.

The substitution-based alternative works properly because condition
evaluation by means of a $\lambda$-complete condition evaluation
operator is not dependent on process behaviour.
Hence, the result of condition evaluation is globally valid.
Below, we will generalize the condition evaluation operators introduced
above in such a way that condition evaluation may be dependent on
process behaviour.
In that case, the result of condition evaluation is in general not
globally valid.

\begin{remark}
\label{remark-theory}
Assume that $\lambda$ is a regular infinite cardinal.
Let $h \in \CondHom{\lambda}$.
Then $h$ induces a theory $\Phi \subset \PropAlg{\lambda}$ such that
$h = h_\Phi$, viz.\ the theory $\Phi$ defined by
\begin{ldispl}
\Phi =
 \set{\repr{h(\alpha)} \Iff \repr{\alpha} \where
      \alpha \in \CondAlg{\lambda}} \union
 \set{\repr{\alpha} \Iff \repr{\beta} \where h(\alpha) = h(\beta)}\;.
\end{ldispl}
Consequently, if $\lambda$ is a regular infinite cardinal, condition
evaluation by means of the $\lambda$-complete condition evaluation
operators introduced above is always condition evaluation of which the
result can be determined from a set of propositions.
We will return to this observation in
Section~\ref{sect-signal-emission}.
\end{remark}

We proceed with generalizing the condition evaluation operators
introduced above.
It is assumed that a fixed but arbitrary function
$\funct{\eff}{\Act \x \CondHom{\lambda}}{\CondHom{\lambda}}$ has been
given.

There is a unary
\emph{generalized $\lambda$-complete condition evaluation} operator
$\funct{\gceval{h}}{\Proc}{\Proc}$ for each $h \in \CondHom{\lambda}$;
and there is again the unary operator $\funct{\ceval{h}}{\Cond}{\Cond}$
for each $h \in \CondHom{\lambda}$.

The $\lambda$-complete generalized condition evaluation operator
$\gceval{h}$ allows, given the function $\eff$, to evaluate conditions
dependent of process behaviour.
The function $\eff$ gives, for each action $a$ and $\lambda$-complete
endomorphism $h$, the $\lambda$-complete endomorphism $h'$ that
represents the changed results of condition evaluation due to performing
$a$.
The function $\eff$ is extended to $\Actd$ such that $\eff(\dead,h) = h$
for all $h \in \CondHom{\lambda}$.

The additional axioms for $\gceval{h}$, where $h \in \CondHom{\lambda}$,
are the axioms given in Table~\ref{axioms-gceval} and axioms CE6--CE11
from Table~\ref{axioms-ceval}.
\begin{table}[!t]
\lcaption{Axioms for generalized condition evaluation}{$a \in \Actd$}
\label{axioms-gceval}
\begin{eqntbl}
\begin{axcol}
\gceval{h}(\ep) = \ep                                   & \axiom{GCE1T}\\
\gceval{h}(a \seqc x) = a \seqc \gceval{\eff(a,h)}(x)   & \axiom{GCE2}\\
\gceval{h}(x \altc y) =
                      \gceval{h}(x) \altc \gceval{h}(y) & \axiom{GCE3}\\
\gceval{h}(\phi \gc x) =
                      \ceval{h}(\phi) \gc \gceval{h}(x) & \axiom{GCE4}
\end{axcol}
\end{eqntbl}
\end{table}

The structural operational semantics of \ACPec\ extended with
generalized condition evaluation is described by the transition rules
for \ACPec\ and the transition rules given in Table~\ref{sos-gceval}.
\begin{table}[!t]
\caption{Transition rules for generalized condition evaluation}
\label{sos-gceval}
\begin{ruletbl}
\RuleC
{\isterm{x}{\phi}}
{\isterm{\gceval{h}(x)}{h(\phi)}}
{h(\phi) \neq \bot}
\qquad
\RuleC
{\astep{x}{\gact{\phi}{a}}{x'}}
{\astep{\gceval{h}(x)}{\gact{h(\phi)}{a}}{\gceval{\eff(a,h)}(x')}}
{h(\phi) \neq \bot}
\end{ruletbl}
\end{table}

We can add both the $\lambda$-complete condition evaluation operators
and the generalized $\lambda$-complete condition evaluation operators
to \ACPec.
However, Proposition~\ref{prop-generalization-gceval} stated below makes
it clear that the latter operators supersede the former operators.

The full splitting bisimulation models of \ACPec\ with condition
evaluation and/or generalized condition evaluation are simply the
expansions of the full splitting bisimulation models $\CPrce_\kappa$ of
\ACPec, for infinite cardinals $\kappa \leq \lambda$, obtained by
associating with each operator $\ceval{h}$ and/or $\gceval{h}$ the
corresponding re-labeling operation on conditional transition systems.
As mentioned before, full splitting bisimulation models with domain
$\CTSe_{\kappa} / {\sbisim}$ for $\kappa > \lambda$ are excluded.

The equation $\ceval{h}(\ceval{h'}(x)) = \ceval{h \comp h'}(x)$ is an
axiom, but the equation
$\gceval{h}(\gceval{h'}(x)) = \gceval{h \comp h'}(x)$ is not an axiom.
The reason is that the latter equation is only valid if $\eff$ satisfies
$\eff(a,h \comp h') = \eff(a,h) \comp \eff(a,h')$ for all $a \in \Act$
and $h,h' \in \CondHom{\lambda}$.

As their name suggests, the generalized $\lambda$-complete condition
evaluation operators are generalizations of the $\lambda$-complete
condition evaluation operators.
\begin{proposition}[Generalization]
\label{prop-generalization-gceval}
We can fix the function $\eff$ such that $\gceval{h}(x) = \ceval{h}(x)$
for all $h \in \CondHom{\lambda}$.
\end{proposition}
\begin{proof}
Clearly, if $\eff(a,h') = h'$ for all $a \in \Act$ and
$h' \in \CondHom{\lambda}$, then $\gceval{h}(x) = \ceval{h}(x)$ for all
$h \in \CondHom{\lambda}$.
\qed
\end{proof}
The $\lambda$-complete state operators that are added to \ACPec\ in
Section~\ref{sect-state-operator} are in their turn generalizations
of the generalized $\lambda$-complete condition evaluation operators.

We come back to the $\lambda$-complete condition evaluation operators
$\ceval{h}$ for $h \in \CondHom{\lambda}$.
The image of $\CondAlg{\lambda}$ under the $\lambda$-complete
endomorphism $h$ is a subalgebra of $\CondAlg{\lambda}$ that is
$\lambda$-complete too.
For that reason, we could have used $\lambda$-complete homomorphisms to
subalgebras that are $\lambda$-complete instead of $\lambda$-complete
endomorphisms.
It would go beyond the models of the theory developed so far to
generalize this in such a way that $\lambda$-complete homomorphisms to
$\lambda$-complete Boolean algebras different from subalgebras of
$\CondAlg{\lambda}$ are also included.

However, in the case where we consider $\lambda$-complete homomorphisms
between free $\lambda$-complete Boolean algebras over different sets of
generators, we can relate the models for different choices for $\ACond$.

Let $C$ and $C'$ be different choices for $\ACond$,%
\footnote
{The interesting cases are those where the cardinalities of $C$ and
 $C'$ are different. Otherwise, the homomorphisms are isomorphisms.}
and let $\CPrce_\kappa(C)$ and $\CPrce_\kappa(C')$, for
$\kappa \leq \lambda$, be the full splitting bisimulation models
$\CPrce_\kappa$ of \ACPec\ for the different choices for $\ACond$.
Moreover, let $h$ be a $\lambda$-complete homomorphism from the free
$\lambda$-complete Boolean algebra over $C$ to the free
$\lambda$-complete Boolean algebra over $C'$.
Then $h$ can be extended to a homomorphism $h^*$ from $\CPrce_\kappa(C)$
to $\CPrce_\kappa(C')$.
This homomorphism is defined by
\begin{ldispl}
h^*(\sbeqvc{\tup{S,{\stepf},{\stermf},s^0}}) =
\sbeqvc{\conn(S,{\stepf}',{\stermf}',s^0)}\;,
\end{ldispl}
where for every
$\tup{\alpha,a} \in \CondAlgm{\kappa} \x \Act$ and
$\alpha' \in \CondAlgm{\kappa}$:
\begin{ldispl}
\begin{aeqns}
{\step{\tup{\alpha,a}}}{}' & = &
\set{\tup{s,s'} \where
     \Exists{\beta}
      {\astep{s}{\gact{\beta}{a}}{s'} \And \alpha = h(\beta)}}\;,
\\
{\sterm{\alpha'}}{}' & = &
\set{s \where
     \Exists{\beta}{\isterm{s}{\beta} \And \alpha' = h(\beta)}}\;.
\end{aeqns}
\end{ldispl}
It is easy to see that $h^*$ is well-defined and a homomorphism indeed.

Thus, a $\lambda$-complete homomorphism between $\lambda$-complete
Boolean algebras over different sets of generators can be used to
translate conditions throughout a full splitting bisimulation model for
one choice of $\ACond$ in such a way that a full splitting bisimulation
model for a different choice of $\ACond$ is obtained.

\section{State Operators}
\label{sect-state-operator}

The state operators make it easy to represent the execution of a
process in a state.
The basic idea is that the execution of an action in a state has effect
on the state, i.e.\ it causes a change of state.
Besides, there is an action left when an action is executed in a state.
The operators introduced here generalize the state operators added to
\ACP\ in~\cite{BB88}.
The main difference with those operators is that guarded commands are
taken into account.
%
As in the case of the condition evaluation operators and the generalized
condition evaluation operators, these state operators require to fix an
infinite cardinal $\lambda$.
By doing so, full splitting bisimulation models with domain
$\CTSe_{\kappa} / {\sbisim}$ for $\kappa > \lambda$ are excluded.

It is assumed that a fixed but arbitrary set $S$ of \emph{states} has
been given, together with functions
$\funct{\act}{\Act \x S}{\Actd}$, $\funct{\eff}{\Act \x S}{S}$ and
$\funct{\eval}{\CondAlg{\lambda} \x S}{\CondAlg{\lambda}}$, where, for
each $s \in S$, the function
$\funct{h_s}{\CondAlg{\lambda}}{\CondAlg{\lambda}}$ defined by
$h_s(\alpha) = \eval(\alpha,s)$ is a $\lambda$-complete endomorphism of
$\CondAlg{\lambda}$.

There are unary $\lambda$-\emph{complete state operators}
$\funct{\state{s}}{\Proc}{\Proc}$ and $\funct{\state{s}}{\Cond}{\Cond}$
for each $s \in S$.%
\footnote
{Holding on to the usual conventions leads to the double use of the
 symbol $\lambda$: without subscript it stands for an infinite cardinal,
 and with subscript it stands for a state operator.}

The $\lambda$-complete state operator $\state{s}$ allows, given the
above-mentioned functions, processes to interact with a state.
Let $p$ be a process.
Then $\state{s}(p)$ is the process $p$ executed in state $s$.
The function $\act$ gives, for each action $a$ and state $s$,
the action that results from executing $a$ in state $s$.
The function $\eff$ gives, for each action $a$ and state $s$,
the state that results from executing $a$ in state $s$.
The function $\eval$ gives, for each condition $\alpha$ and state $s$,
the condition that results from evaluating $\alpha$ in state $s$.
The functions $\act$ and $\eff$ are extended to $\Actd$ such that
$\act(\dead,s) = \dead$ and $\eff(\dead,s) = s$ for all $s \in S$.

The additional axioms for $\state{s}$, where $s \in S$, are the axioms
given in Table~\ref{axioms-state-op}.
\begin{table}[!t]
\lcaption{Axioms for state operators}
  {$a \in \Actd$, $\eta \in \ACond$,
   $\eta' \in \ACond \union \set{\bot,\top}$}
\label{axioms-state-op}
\begin{eqntbl}
\begin{axcol}
\state{s}(\ep) = \ep                                   & \axiom{SO1T} \\
\state{s}(a \seqc x) = \act(a,s) \seqc \state{\eff(a,s)}(x)
                                                       & \axiom{SO2} \\
\state{s}(x \altc y) = \state{s}(x) \altc \state{s}(y) & \axiom{SO3} \\
\state{s}(\phi \gc x) = \state{s}(\phi) \gc \state{s}(x)
                                                       & \axiom{SO4}
\end{axcol}
\quad\quad
\begin{axcol}
\state{s}(\bot) = \bot                                 & \axiom{SO5} \\
\state{s}(\top) = \top                                 & \axiom{SO6} \\
\state{s}(\eta) = \eta'
                     \qquad \mif \eval(\eta,s) = \eta' & \axiom{SO7} \\
\state{s}(\bcompl \phi) = \bcompl \state{s}(\phi)      & \axiom{SO8} \\
\state{s}(\phi \join \psi) =
                 \state{s}(\phi) \join \state{s}(\psi) & \axiom{SO9} \\
\state{s}(\phi \meet \psi) =
                 \state{s}(\phi) \meet \state{s}(\psi) & \axiom{SO10}
\end{axcol}
\end{eqntbl}
\end{table}

The structural operational semantics of \ACPec\ extended with
state operators is described by the transition rules for \ACPec\
and the transition rules given in Table~\ref{sos-state-op}.
\begin{table}[!t]
\caption{Transition rules for state operators}
\label{sos-state-op}
\begin{ruletbl}
\RuleC
{\isterm{x}{\phi}}
{\isterm{\state{s}(x)}{\eval(\phi,s)}}
{\eval(\phi,s) \neq \bot}
\\
\RuleC
{\astep{x}{\gact{\phi}{a}}{x'}}
{\astep{\state{s}(x)}{\gact{\eval(\phi,s)}{\act(a,s)}}
                               {\state{\eff(a,s)}(x')}}
{\act(a,s) \neq \dead,\; \eval(\phi,s) \neq \bot}
\end{ruletbl}
\end{table}

The full splitting bisimulation models of \ACPec\ with state operators
are simply the expansions of the full splitting bisimulation models
$\CPrce_\kappa$ of \ACPec\ obtained by associating with each operator
$\state{s}$ the corresponding re-labeling operation on conditional
transition systems.

We can add, in addition to the $\lambda$-complete state operators, the
$\lambda$-complete condition evaluation operators and/or the generalized
$\lambda$-complete condition evaluation operators from
Section~\ref{sect-cond-eval} to \ACPec.

We write $\CPrceext_\kappa$ for the expansion of $\CPrce_\kappa$ for the
$\lambda$-complete condition evaluation operators, the generalized
$\lambda$-complete condition evaluation operators and the
$\lambda$-complete state operators.

The $\lambda$-complete state operators are generalizations of the
generalized $\lambda$-complete condition evaluation operators from
Section~\ref{sect-cond-eval}.
\begin{proposition}[Generalization]
\label{prop-generalization-so}
We can fix $S$, $\act$, $\eff$ and $\eval$ such that, for some
$\funct{f}{\CondHom{\lambda}}{S}$,\, $\state{f(h)}(x) = \gceval{h}(x)$
holds for all $h \in \CondHom{\lambda}$ in all full splitting
bisimulation models $\CPrceext_\kappa$ with $\kappa \leq \lambda$.
\end{proposition}
\begin{proof}
Clearly, if $S = \CondHom{\lambda}$, $f$ is the identity function on
$\CondHom{\lambda}$, and $\act(a,s) = a$,
$\eff(a,s) = \eff(a,f^{-1}(s))$ and
$\eval(\alpha,s) = f^{-1}(s)(\alpha)$ for all $a \in \Act$, $s \in S$
and $\alpha \in \CondAlg{\lambda}$, then
$\state{f(h)}(x) = \gceval{h}(x)$ holds for all
$h \in \CondHom{\lambda}$ in all full splitting bisimulation models
$\CPrceext_\kappa$ with $\kappa \leq \lambda$.
\qed
\end{proof}

\section{Signal Emission}
\label{sect-signal-emission}

In Section~\ref{sect-cond-eval}, we made the observation that, if
$\lambda$ is a regular infinite cardinal, condition evaluation by means
of the $\lambda$-complete condition evaluation operators $\ceval{h}$
from that section is always condition evaluation of which the result can
be determined from a set of propositions
(see Remark~\ref{remark-theory}).
A similar observation can be made about condition evaluation by means
of the generalized $\lambda$-complete condition evaluation operators
$\gceval{h}$ from that section.
In the case of condition evaluation by means of $\ceval{h}$, the set of
propositions determining the result of condition evaluation does not
change as a process proceeds.
In the case of condition evaluation by means of $\gceval{h}$, it may
happen that the set of propositions determining the result of condition
evaluation changes as a process proceeds.
That is, the sets of propositions relevant to a process and its
subprocesses may differ.
This suggest that condition evaluation can also be dealt with by
explicitly associating sets of propositions with processes.
The intuition is, then, that all propositions from the set of
propositions associated with a process holds at the start of the
process.

Clearly, if we restrict ourselves to sets of propositions of cardinality
less than a regular infinite cardinal $\lambda$, we can associate
elements of $\CondAlg{\lambda}$ with processes instead.
In line with~\cite{BB92c}, the element of $\CondAlg{\lambda}$ associated
with a process is called the signal emitted by the process.
Because $\bot$ represents the proposition $\False$, the proposition
that cannot hold at the start of any process, we regard a process with
which $\bot$ is associated as an inconsistency.
However, in an algebraic setting, we cannot exclude this inconsistency.
Therefore, we consider it to be a special process, which is called the
inaccessible process.%
\footnote
{In~\cite{BMU98a,BM03a}, this process is rather contradictory called
 the non-existent process. Its new name was prompted by the fact that
 after performing an action no process will ever proceed as this
 process.}

The idea to associate elements of $\CondAlg{\lambda}$ with processes
naturally suggests itself in the case where $\lambda$ is a regular
infinite cardinal.
However, there are no trammels to drop the restriction that $\lambda$ is
regular.

All this leads us to an extension of \ACPec, called \ACPecs, with the
following additional constants and operators:
\begin{iteml}
\item
the \emph{inaccessible process} constant $\const{\nex}{\Proc}$;
\item
the binary \emph{signal emission} operator
$\funct{\emi}{\Cond \x \Proc}{\Proc}$.
\end{iteml}

The axioms of \ACPecs\ are the axioms of \ACPec\ with axioms CM2, CM3
and GC8--GC10 replaced by axioms CM2ST, CM3S and GC8S--GC10S from
Table~\ref{axioms-ACPec-adapt-emi},%
\begin{table}[!t]
\lcaption{Axioms adapted to signal emission}{$a \in \Actd$}
\label{axioms-ACPec-adapt-emi}
\begin{eqntbl}
\begin{axcol}
\ep \leftm x = \encap{\Act}(x)                       & \axiom{CM2ST}\\
a \seqc x \leftm y =
           a \seqc (x \parc y) \altc \encap{\Act}(y) & \axiom{CM3S}
\eqnsep
(\phi \gc x) \leftm y =
         \phi \gc (x \leftm y) \altc \encap{\Act}(y) & \axiom{GC8S}\\
(\phi \gc x) \commm y =
         \phi \gc (x \commm y) \altc \encap{\Act}(y) & \axiom{GC9S}\\
x \commm (\phi \gc y) =
         \phi \gc (x \commm y) \altc \encap{\Act}(x) & \axiom{GC10S}
\end{axcol}
\end{eqntbl}
\end{table}
and the additional axioms given in Table~\ref{axioms-ACPec-emi}.%
\begin{table}[!t]
\lcaption{Additional axioms for signal emission}{$a \in \Actd$}
\label{axioms-ACPec-emi}
\begin{eqntbl}
\begin{axcol}
x \altc \nex = \nex                                   & \axiom{NE1}\\
\nex \seqc x = \nex                                   & \axiom{NE2}\\
a \seqc \nex = \dead                                  & \axiom{NE3}\\
\top \emi x = x                                       & \axiom{SE1}\\
\bot \emi x = \nex                                    & \axiom{SE2}\\
\phi \emi x \altc y = \phi \emi (x \altc y)           & \axiom{SE3}\\
(\phi \emi x) \seqc y = \phi \emi x \seqc y           & \axiom{SE4}
\end{axcol}
\qquad
\begin{axcol}
\phi \emi (\psi \emi x ) = (\phi \meet \psi) \emi x   & \axiom{SE5}\\
\phi \emi (\phi \gc x) = \phi \emi  x                 & \axiom{SE6}\\
\phi \gc (\psi \emi x) =
          (\bcompl \phi \join \psi) \emi (\phi \gc x) & \axiom{SE7}\\
(\phi \emi x) \leftm y = \phi \emi (x \leftm y)       & \axiom{SE8}\\
(\phi \emi x) \commm y = \phi \emi (x \commm y)       & \axiom{SE9}\\
x \commm (\phi \emi y) = \phi \emi (x \commm y)       & \axiom{SE10}\\
\encap{H}(\phi \emi x) = \phi \emi \encap{H}(x)       & \axiom{SE11}
\end{axcol}
\end{eqntbl}
\end{table}
Axioms NE1--NE3 and SE1--SE11 have been taken from~\cite{BB94b} and
axioms GC9S and GC10S have been taken from~\cite{BB94b} with subterms
of the form $\rsgn{x} \emi \dead$ replaced by $\encap{\Act}(x)$.
Axioms CM2ST, CM3S and GC8S differ really from the corresponding axioms
in~\cite{BB94b} due to the choice of having as the signal emitted by
the left merge of two processes, as in the case of the communication
merge, always the meet of the signals emitted by the two processes.

In the structural operational semantics of \ACPecs, unary relations
$\sgn{\alpha}$, one for each $\alpha \in \FinCondAlg \diff \set{\bot}$,
are used in addition to the relations $\sterm{\alpha}$ and $\step{\ell}$.
We write $\rsgn{p} = \alpha$ instead of $p \in \sgn{\alpha}$.
The relation $\sgn{\alpha}$ can be explained as follows:
\begin{iteml}
\item
$\rsgn{p} = \alpha$: $p$ emits the signal $\alpha$.
\end{iteml}
The structural operational semantics of \ACPecs\ is described by the
transition rules given in Table~\ref{sos-ACPec-emi}.%
\begin{table}[!t]
\caption{Transition rules for \ACPecs}
\label{sos-ACPec-emi}
\begin{ruletbl}
\Rule
{\phantom{\isterm{\ep}{\top}}}
{\isterm{\ep}{\top}}
\qquad
\Rule
{\phantom{\astep{a}{\gact{\top}{a}}{\ep}}}
{\astep{a}{\gact{\top}{a}}{\ep}}
\\
\Rule
{\isterm{x}{\phi},\; \rsgn{x \altc y} \neq \bot}
{\isterm{x \altc y}{\phi}}
\quad
\Rule
{\isterm{y}{\phi},\; \rsgn{x \altc y} \neq \bot}
{\isterm{x \altc y}{\phi}}
\\
\Rule
{\astep{x}{\gact{\phi}{a}}{x'},\; \rsgn{x \altc y} \neq \bot}
{\astep{x \altc y}{\gact{\phi}{a}}{x'}}
\quad
\Rule
{\astep{y}{\gact{\phi}{a}}{y'},\; \rsgn{x \altc y} \neq \bot}
{\astep{x \altc y}{\gact{\phi}{a}}{y'}}
\\
\Rule
{\isterm{x}{\phi},\; \isterm{y}{\phi}}
{\isterm{x \seqc y}{\phi \meet \psi}}
\quad
\RuleC
{\isterm{x}{\phi},\; \astep{y}{\gact{\psi}{a}}{y'}}
{\astep{x \seqc y}{\gact{\phi \meet \psi}{a}}{y'}}
{\phi \meet \psi \neq \bot}
\quad
\Rule
{\astep{x}{\gact{\phi}{a}}{x'}}
{\astep{x \seqc y}{\gact{\phi}{a}}{x' \seqc y}}
\\
\RuleC
{\isterm{x}{\phi}}
{\isterm{\psi \gc x}{\phi \meet \psi}}
{\phi \meet \psi \neq \bot}
\quad
\RuleC
{\astep{x}{\gact{\phi}{a}}{x'}}
{\astep{\psi \gc x}{\gact{\phi \meet \psi}{a}}{x'}}
{\phi \meet \psi \neq \bot}
\\
\Rule
{\isterm{x}{\phi},\; \rsgn{\psi \emi x} \neq \bot}
{\isterm{\psi \emi x}{\phi}}
\quad
\Rule
{\astep{x}{\gact{\phi}{a}}{x'},\; \rsgn{\psi \emi x} \neq \bot}
{\astep{\psi \emi x}{\gact{\phi}{a}}{x'}}
\\
\RuleC
{\isterm{x}{\phi},\; \isterm{y}{\psi},\; \rsgn{x \parc y} \neq \bot}
{\isterm{x \parc y}{\phi \meet \psi}}
{\phi \meet \psi \neq \bot}
\\
\Rule
{\astep{x}{\gact{\phi}{a}}{x'},\;
 \rsgn{x \parc y} \neq \bot,\; \rsgn{x' \parc y} \neq \bot}
{\astep{x \parc y}{\gact{\phi}{a}}{x' \parc y}}
\quad
\Rule
{\astep{y}{\gact{\phi}{a}}{y'},\;
 \rsgn{x \parc y} \neq \bot,\; \rsgn{x \parc y'} \neq \bot}
{\astep{x \parc y}{\gact{\phi}{a}}{x \parc y'}}
\\
\RuleC
{\astep{x}{\gact{\phi}{a}}{x'},\; \astep{y}{\gact{\psi}{b}}{y'},\;
 \rsgn{x \parc y} \neq \bot,\; \rsgn{x' \parc y'} \neq \bot}
{\astep{x \parc y}{\gact{\phi \meet \psi}{c}}{x' \parc y'}}
{a \commm b = c,\; \phi \meet \psi \neq \bot}
\\
\Rule
{\astep{x}{\gact{\phi}{a}}{x'},\;
 \rsgn{x \leftm y} \neq \bot,\; \rsgn{x' \parc y} \neq \bot}
{\astep{x \leftm y}{\gact{\phi}{a}}{x' \parc y}}
\\
\RuleC
{\astep{x}{\gact{\phi}{a}}{x'},\; \astep{y}{\gact{\psi}{b}}{y'},\;
 \rsgn{x \commm y} \neq \bot,\; \rsgn{x' \parc y'} \neq \bot}
{\astep{x \commm y}{\gact{\phi \meet \psi}{c}}{x' \parc y'}}
{a \commm b = c,\; \phi \meet \psi \neq \bot}
\\
\Rule
{\isterm{x}{\phi}}
{\isterm{\encap{H}(x)}{\phi}}
\quad
\RuleC
{\astep{x}{\gact{\phi}{a}}{x'}}
{\astep{\encap{H}(x)}{\gact{\phi}{a}}{\encap{H}(x')}}
{a \not\in H}
\eqnsep
\Rule
{\phantom{\rsgn{\nex} = \bot}}
{\rsgn{\nex} = \bot}
\quad
\Rule
{\phantom{\rsgn{\ep} = \top}}
{\rsgn{\ep} = \top}
\quad
\Rule
{\phantom{\rsgn{a} = \top}}
{\rsgn{a} = \top}
\quad
\Rule
{\rsgn{x} = \phi,\; \rsgn{y} = \psi}
{\rsgn{x \altc y} = \phi \meet \psi}
\\
\Rule
{\rsgn{x} = \phi}
{\rsgn{x \seqc y} = \phi}
\quad
\Rule
{\rsgn{x} = \phi}
{\rsgn{\psi \gc y} = \bcompl{\psi} \join \phi}\quad
\Rule
{\rsgn{x} = \phi}
{\rsgn{\psi \emi y} = \psi \meet \phi}
\\
\Rule
{\rsgn{x} = \phi,\; \rsgn{y} = \psi}
{\rsgn{x \parc y} = \phi \meet \psi}
\quad
\Rule
{\rsgn{x} = \phi,\; \rsgn{y} = \psi}
{\rsgn{x \leftm y} = \phi \meet \psi}
\quad
\Rule
{\rsgn{x} = \phi,\; \rsgn{y} = \psi}
{\rsgn{x \commm y} = \phi \meet \psi}
\quad
\Rule
{\rsgn{x} = \phi}
{\rsgn{\encap{H}(x)} = \phi}
\end{ruletbl}
\end{table}
These transition rules include all transition rules from
Table~\ref{sos-ACPec} with additional premises to exclude transitions
from or to processes that emit the signal $\bot$.
There are additional transition rules describing the signals emitted by
the processes.
The transition rules for signal emission are new as well.

The following gives a good picture of the nature of signals and
conditions.
\begin{proposition}[Signals and conditions]
\label{prop-sgn-cond}
If $\Der{\repr{\alpha}\;}{\;\repr{\beta} {\Iff} \repr{\beta'}}$, then
$\cterm{\alpha} \emi (\cterm{\beta} \gc x) =
 \cterm{\alpha} \emi (\cterm{\beta'} \gc x)$.
\end{proposition}
\begin{proof}
The proof is the same to the proof of the corresponding proposition in
the setting of \ACPcs\ given in~\cite{BM05a}.
\qed
\end{proof}
We have the following corollaries from Proposition~\ref{prop-sgn-cond}.
\begin{corollary}
\label{corollary-sgn-cond-1}
If $\Der{\repr{\alpha}\;}{\;\repr{\beta}}$, then
$\cterm{\alpha} \emi (\cterm{\beta} \gc x) = \cterm{\alpha} \emi x$.
If $\Der{\repr{\alpha}\;}{\;\Not \repr{\beta}}$, then
$\cterm{\alpha} \emi (\cterm{\beta} \gc x) = \cterm{\alpha} \emi \dead$.
\end{corollary}
\begin{corollary}
\label{corollary-sgn-cond-2}
If $\eff(h,a)$ is the identity endomorphism on $\FinCondAlg$ for
all endomorphisms $h$ on $\FinCondAlg$ and $a \in \Act$, then we have
$\gceval{h_{\set{\repr{\alpha}}}}(\cterm{\beta} \gc x) =
 \cterm{\beta'} \gc \gceval{h_{\set{\repr{\alpha}}}}(x)$
implies
$\cterm{\alpha} \emi (\cterm{\beta} \gc x) =
 \cterm{\alpha} \emi (\cterm{\beta'} \gc x)$.
\end{corollary}

\section{Full Signal-Observing Splitting Bisimulation Models of \ACPecs}
\label{sect-full-bisim-ACPecs}

In this section, we introduce conditional transition systems with
signals, signal-observing splitting bisimilarity of conditional
transition systems with signals, and the full signal-observing splitting
bisimulation models of \ACPecs.

Conditional transition systems with signals generalize conditional
transition systems.

Let $\kappa$ be an infinite cardinal.
Then a $\kappa$-\emph{conditional transition system with signals} $T$
is a tuple $\tup{S,{\stepf},{\stermf},{\sgnf},s^0}$ where
\begin{iteml}
\item
$\tup{S,{\stepf},{\stermf},s^0}$ is a $\kappa$-conditional transition
system;
\item
${\sgnf}$ is a function from $S$ to $\CondAlg{\kappa}$;
\end{iteml}
and for all $\ell \in \CondAlgm{\kappa} \x \Act$
and $\alpha \in \CondAlgm{\kappa}$:
\begin{iteml}
\item
$\set{\tup{s,s'} \in {\step{\ell}} \where
      \rsgn{s} = \bot \Or \rsgn{s'} = \bot} = \emptyset$;
\item
$\set{s \in {\sterm{\alpha}} \where \rsgn{s} = \bot} = \emptyset$.
\end{iteml}
We say that $\rsgn{s}$ is the signal emitted by the state $s$.

For conditional transition systems with signals, reachability and
connectedness are defined exactly as for conditional transition
systems.

Let $\tup{S,{\stepf},{\stermf},{\sgnf},s^0}$ be a $\kappa$-conditional
transition system with signals (for an infinite cardinal $\kappa$) that
is not necessarily connected.
Then the \emph{connected part} of $T$, written $\conn(T)$, is simply
defined as follows:
\begin{ldispl}
\conn(T) = \tup{S',{\stepf}',{\stermf}',{\sgnf}',s^0}\;,
\end{ldispl}
where
\begin{ldispl}
\tup{S',{\stepf}',{\stermf}',s^0} =
\conn\tup{S,{\stepf},{\stermf},s^0}\;,
\\
\mbox{${\sgnf}'$ is the restriction of ${\sgnf}$ to $S'$}\;.
\end{ldispl}

Let $\kappa$ be an infinite cardinal.
Then $\CTSes_\kappa$ is the set of all $\kappa$-conditional transition
systems with signals $\tup{S,{\stepf},{\stermf},{\sgnf},s^0}$ for
which $\tup{S,{\stepf},{\stermf},s^0} \in \CTSe_\kappa$.

Isomorphism between conditional transition systems with signals is
defined as between conditional transition systems, but with the
additional condition that $\rsgni{s} = \rsgnii{b(s)}$.
Splitting bisimilarity has to be adapted to the setting with signals.

Let $T_1 = \tup{S_1,{\stepf}_1,{\stermf}_1,{\sgnf}_1,s^0_1}
     \in \CTSes_\kappa$,
    $T_2 = \tup{S_2,{\stepf}_2,{\stermf}_2,{\sgnf}_2,s^0_2}
     \in \CTSes_\kappa$
(for an infinite cardinal $\kappa$).
Then a \emph{signal-observing splitting bisimulation} $B$ between $T_1$
and $T_2$ is a binary relation $B \subseteq S_1 \x S_2$ such that
$B(s^0_1,s^0_2)$ and for all $s_1,s_2$ such that $B(s_1,s_2)$:
\begin{iteml}
\item
$\rsgni{s_1} = \rsgnii{s_2}$;
\item
if $\astepi{s_1}{\gact{\alpha}{a}}{s_1'}$, then there is a set
$CS_2' \subseteq \CondAlgm{\kappa} \x S_2$ of cardinality less than
$\kappa$ such that
$\rsgni{s_1} \meet \alpha \beloweq \infjoin \dom(CS_2')$ and
for all $\tup{\alpha',s_2'} \in CS_2'$,
$\astepii{s_2}{\gact{\alpha'}{a}}{s_2'}$ and $B(s_1',s_2')$;
\item
if $\astepii{s_2}{\gact{\alpha}{a}}{s_2'}$, then there is a set
$CS_1' \subseteq \CondAlgm{\kappa} \x S_1$ of cardinality less than
$\kappa$ such that
$\rsgnii{s_2} \meet \alpha \beloweq \infjoin \dom(CS_1')$ and
for all $\tup{\alpha',s_1'} \in CS_1'$,
$\astepi{s_1}{\gact{\alpha'}{a}}{s_1'}$ and $B(s_1',s_2')$;
\item
if $\istermi{s_1}{\alpha}$, then there is a set
$C' \subseteq \CondAlgm{\kappa}$ of cardinality less than $\kappa$ such
that $\rsgni{s_1} \meet \alpha \beloweq \infjoin C'$
and for all $\alpha' \in C'$, $\istermii{s_2}{\alpha'}$;
\item
if $\istermii{s_2}{\alpha}$, then there is a set
$C' \subseteq \CondAlgm{\kappa}$ of cardinality less than $\kappa$ such
that $\rsgnii{s_2} \meet \alpha \beloweq \infjoin C'$
and for all $\alpha' \in C'$, $\istermi{s_1}{\alpha'}$.
\end{iteml}
Two conditional transition systems with signals
$T_1,T_2 \in \CTSes_\kappa$ are
\emph{signal-ob\-serving splitting bisimilar}, written
$T_1 \ssbisim T_2$, if there exists a signal-observing splitting
bisimulation $B$ between $T_1$ and $T_2$.
Let $B$ be a signal-observing splitting bisimulation between $T_1$ and
$T_2$.
Then we say that $B$ is a splitting signal-observing bisimulation
\emph{witnessing} $T_1 \ssbisim T_2$.

It is straightforward to see that ${\ssbisim}$ is an equivalence on
$\CTSes_\kappa$.
Let $T \in \CTSes_\kappa$.
Then we write $\ssbeqvc{T}$ for
$\set{T' \in \CTSes_\kappa \where T \ssbisim T'}$, i.e.\ the
${\ssbisim}$\,-equivalence class of $T$.
We write $\CTSes_\kappa / {\ssbisim}$ for the set of equivalence classes
$\set{\ssbeqvc{T} \where T \in \CTSes_\kappa}$.

The elements of $\CTSes_\kappa$ and operations on $\CTSes_\kappa$ to be
associated with the constants and operators of \ACPec\ are as the
elements of $\CTSe_\kappa$ and operations on $\CTSe_\kappa$ associated
with them, but with all relations $\sterm{\alpha}$ and $\step{\ell}$
restricted to states that emit a signal different from $\bot$ and
with the additional function ${\sgn{}}$ as suggested by the structural
operational semantics of \ACPecs.

We associate with the additional constant $\nex$ an element $\tsnexs$
of $\CTSes_\kappa$ and with the additional operator $\emi$
an operation $\tsemis$ on $\CTSes_\kappa$ as follows.
\begin{iteml}
\item
\begin{ldispl}
\tsnexs = \tup{\set{s^0},\emptyset,\emptyset,{\sgn{}},s^0}\;,
\end{ldispl}
where
\begin{ldispl}
\rsgn{s^0} = \bot\;.
\end{ldispl}
\item
Let $T = \tup{S,{\stepf},{\stermf},{\sgnf},s^0} \in \CTSes_\kappa$.
Then
\begin{ldispl}
\alpha \tsemis T = \conn(S,{\stepf}',{\stermf}',{\sgnf}',s^0)\;,
\end{ldispl}
where
\begin{ldispl}
\begin{aeqns}
\rsgnp{s}   & = & \rsgn{s} &
                         \quad \mathrm{for}\; s \in S \diff \set{s^0}\;,
\\
\rsgnp{s^0} & = & \alpha \meet \rsgn{s^0}\;,
\end{aeqns}
\end{ldispl}
and for every
$\tup{\alpha,a} \in \CondAlgm{\kappa} \x \Act$ and
$\alpha' \in \CondAlgm{\kappa}$:
\begin{ldispl}
\begin{aeqns}
{\step{\tup{\alpha,a}}}{}' & = &
\set{\tup{s,s'} \where \astep{s}{\gact{\alpha}{a}}{s'} \And
                       \rsgnp{s} \neq \bot \And \rsgnp{s'} \neq \bot}\;,
\\
{\sterm{\alpha'}}{}' & = &
\set{s \where \isterm{s}{\alpha'} \And \rsgnp{s} \neq \bot}\;.
\end{aeqns}
\end{ldispl}
\end{iteml}

We can easily show that signal-observing splitting bisimilarity is a
congruence with respect to the operations on $\CTSes_\kappa$ associated
with the operators of \ACPecs.
\begin{proposition}[Congruence]
\label{prop-congruence-ACPecs}
Let $\kappa$ be an infinite cardinal.
Then for all $T_1,T_2,T_1',T_2' \in \CTSe_\kappa$ and
$\alpha \in \CondAlg{\kappa}$,
$T_1 \ssbisim T_1'$ and $T_2 \ssbisim T_2'$ imply
$T_1 \tsaltcs T_2 \ssbisim T_1' \tsaltcs T_2'$,
$T_1 \tsseqcs T_2 \ssbisim T_1' \tsseqcs T_2'$,
$\alpha \tsgcs T_1 \ssbisim \alpha \tsgcs T_1'$,
$\alpha \tsemis T_1 \ssbisim \alpha \tsemis T_1'$,
$T_1 \tsparcs T_2 \ssbisim T_1' \tsparcs T_2'$,
$T_1 \tsleftms T_2 \ssbisim T_1' \tsleftms T_2'$,
$T_1 \tscommms T_2 \ssbisim T_1' \tscommms T_2'$ and
$\tsencaps{H}(T_1) \ssbisim \tsencaps{H}(T_1')$.
\end{proposition}
\begin{proof}
For $\tsaltcs$, $\tsseqcs$, $\tsgcs$, $\tsparcs$, $\tsleftms$,
$\tscommms$ and $\tsencaps{H}$, witnessing signal-observing splitting
bisimulations are constructed in the same way as witnessing splitting
bisimulations are constructed in the proof of
Proposition~\ref{prop-congruence-ACPec}.
What remains is to construct a witnessing signal-observing splitting
bisimulation for $\tsemis$.
Let $R$ be a signal-observing splitting bisimulation witnessing
$T_1 \ssbisim T_1'$.
Then we construct a relation $R_{\tsemis}$ as follows:
\begin{iteml}
\item
$R_{\tsemis} = R  \inter (S \x S')$,
where $S$ and $S'$ are the sets of states of $\alpha \tsemis T_1$ and
$\alpha \tsemis T_1'$, respectively.
\end{iteml}
Given the definition of signal emission, it is easy to see that
$R_{\tsemis}$ is a signal-observing splitting bisimulation witnessing
$\alpha \tsemis T_1 \ssbisim \alpha \tsemis T_1'$.
\qed
\end{proof}

The ingredients of the
\emph{full signal-observing splitting bisimulation models}
$\CPrces_\kappa$ of \ACPecs, one for each infinite cardinal $\kappa$,
are defined as follows:
\begin{ldispl}
\begin{aeqns}
\cP & = & \CTSes_\kappa / {\ssbisim}\;,
\eqnsep
\snexs & = & \ssbeqvc{\tsnexs}\;,
\eqnsep
\sdeads & = & \ssbeqvc{\tsdeads}\;,
\eqnsep
\seps & = & \ssbeqvc{\tseps}\;,
\eqnsep
\sacts{a} & = & \ssbeqvc{\tsacts{a}}\;,
\eqnsep
\ssbeqvc{T_1} \saltcs \ssbeqvc{T_2} & = & \ssbeqvc{T_1 \tsaltcs T_2}\;,
\eqnsep
\ssbeqvc{T_1} \sseqcs \ssbeqvc{T_2} & = & \ssbeqvc{T_1 \tsseqcs T_2}\;,
\end{aeqns}
\qquad
\begin{aeqns}
{} \eqnsep
\alpha \sgcs \ssbeqvc{T_1} & = & \ssbeqvc{\alpha \tsgcs T_1}\;,
\eqnsep
\alpha \semis \ssbeqvc{T_1} & = & \ssbeqvc{\alpha \tsemis T_1}\;,
\eqnsep
\ssbeqvc{T_1} \sparcs \ssbeqvc{T_2} & = & \ssbeqvc{T_1 \tsparcs T_2}\;,
\eqnsep
\ssbeqvc{T_1} \sleftms \ssbeqvc{T_2} & = & \ssbeqvc{T_1 \tsleftms T_2}\;,
\eqnsep
\ssbeqvc{T_1} \scommms \ssbeqvc{T_2} & = & \ssbeqvc{T_1 \tscommms T_2}\;,
\eqnsep
\sencaps{H}(\ssbeqvc{T_1}) & = & \ssbeqvc{\tsencaps{H}(T_1)}\;.
\end{aeqns}
\end{ldispl}
The operations on $\CTSes_\kappa / {\ssbisim}$ are well-defined because
${\ssbisim}$ is a congruence with respect to the corresponding
operations on $\CTSes_\kappa$.

The structures $\CPrces_\kappa$ are models of \ACPecs.
\begin{theorem}[Soundness of \ACPecs]
\label{theorem-soundness-ACPecs}
For each infinite cardinal $\kappa$, we have
$\Sat{\CPrces_\kappa}{\ACPecs}$.
\end{theorem}
\begin{proof}
Because ${\ssbisim}$ is a congruence, it is sufficient to show that all
axioms are sound.
The soundness of all axioms follows straightforwardly from the
definition of $\CPrces_\kappa$.
\qed
\end{proof}
For all axioms that are in common with \ACPcs, the proof of soundness
with respect to $\CPrces_\kappa$ follows the same line as the proof of
soundness with respect to $\CPrcs_\kappa$.

\section{\ACPe\ with Retrospective Conditions}
\label{sect-ACPecr}

In this section, we present an extension of \ACPec\ with a retrospection
operator on conditions.
The retrospection operator allows for looking back on conditions under
which preceding actions have been performed.
The extension of \ACPec\ with the retrospection operator is called
\ACPecr.

\ACPecr\ has the constants and operators of \ACPec\ and in addition:
\begin{iteml}
\item
the unary \emph{retrospection} operator $\funct{\retro}{\Cond}{\Cond}$;
\item
the unary \emph{retrospection shift} operator
$\funct{\rshiftz}{\Proc}{\Proc}$;
\item
for each $n \in \Nat$, the unary \emph{restricted retrospection shift}
operator $\funct{\rshift{n}}{\Proc}{\Proc}$;
\item
for each $n \in \Nat$, the unary \emph{restricted retrospection shift}
operator $\funct{\rshift{n}}{\Cond}{\Cond}$.
\end{iteml}
In the parallel composition of two processes, when an action of one of
the processes is performed, the retrospections of the other process that
are not internal should go one step further.
This is accomplished by the retrospection shift operator.
The restricted retrospection shift operators, on processes and
conditions, are needed for the axiomatization of the retrospection shift
operator.
The retrospection shift operator $\rshiftz$ is similar to the history
pointer shift operator $\mathit{hps}$ from~\cite{BB98a}.

The axioms of \ACPecr\ are the axioms of \ACPec\ with axiom CM3
replaced by axiom CM3R from Table~\ref{axioms-ACPec-adapt-retro},%
\begin{table}[!t]
\lcaption{Axioms adapted to retrospection}{$a \in \Actd$}
\label{axioms-ACPec-adapt-retro}
\begin{eqntbl}
\begin{axcol}
a \seqc x \leftm y = a \seqc (x \parc \rshiftz(y))      & \axiom{CM3R}
\end{axcol}
\end{eqntbl}
\end{table}
and the additional axioms for retrospection given in
Table~\ref{axioms-ACPecr}.%
\begin{table}[!t]
\lcaption{Additional axioms for retrospection}
  {$a \in \Actd$, $\eta \in \ACond$}
\label{axioms-ACPecr}
\begin{eqntbl}
\begin{axcol}
\retro \bot = \bot                                       & \axiom{R1} \\
\retro \top = \top                                       & \axiom{R2} \\
\retro (\bcompl \phi) = \bcompl (\retro \phi)            & \axiom{R3} \\
\retro (\phi \join \psi) = \retro \phi \join \retro \psi & \axiom{R4} \\
\retro (\phi \meet \psi) = \retro \phi \meet \retro \psi & \axiom{R5} \\
a \seqc (\retro \phi \gc x) = {} \\ \quad
 \phi \gc a \seqc x \altc \bcompl \phi \gc a \seqc \dead & \axiom{R6} \\
{}                                                                    \\
{}                                                                    \\
\rshiftz(x) = \rshift{0}(x)                              & \axiom{RS0} \\
\rshift{n}(\ep) = \ep                                    & \axiom{RS1T}
\end{axcol}
\qquad
\begin{axcol}
\rshift{n}(a \seqc x) = a \seqc \rshift{n+1}(x)          & \axiom{RS2} \\
\rshift{n}(x \altc y) =
                       \rshift{n}(x) \altc \rshift{n}(y) & \axiom{RS3} \\
\rshift{n}(\phi \gc x) =
                      \rshift{n}(\phi) \gc \rshift{n}(x) & \axiom{RS4} \\
\rshift{n}(\bot) = \bot                                  & \axiom{RS5}\\
\rshift{n}(\top) = \top                                  & \axiom{RS6}\\
\rshift{n}(\eta) = \eta                                  & \axiom{RS7}\\
\rshift{n}(\bcompl \phi) = \bcompl \rshift{n}(\phi)      & \axiom{RS8}\\
\rshift{n}(\phi \join \psi) =
                 \rshift{n}(\phi) \join \rshift{n}(\psi) & \axiom{RS9}\\
\rshift{n}(\phi \meet \psi) =
                 \rshift{n}(\phi) \meet \rshift{n}(\psi) & \axiom{RS10}\\
\rshift{0}(\retro \phi) = \retro (\retro \phi)           & \axiom{RS11}\\
\rshift{n+1}(\retro \phi) = \retro \rshift{n}(\phi)      & \axiom{RS12}
\end{axcol}
\end{eqntbl}
\end{table}
The crucial axiom is R6, which shows that a conditional expression of
the form $\retro \zeta \gc p$ gives a retrospection at the condition
under which the immediately preceding action has been performed.
Axiom CM3R shows that retrospections are adapted if two processes
proceed in parallel.
Axioms RS0, RS1T and RS2--RS12 state that this happens as explained
above.
By means of axioms RS5--RS12, the retrospection shift operators on
conditions can be eliminated from all terms of sort~$\Cond$.

Recall that we write $\cond{p}{\zeta}{q}$ for
$\zeta \gc p \altc \bcompl \zeta \gc q$.
An interesting equation is
$a \seqc (\cond{x}{\retro \phi}{y}) =
 \cond{a \seqc x}{\phi}{a \seqc y}$.
This equation is a generalization of axiom R6: axiom R6 is derivable
from the other axioms of \ACPecr\ and this equation by substituting
$\dead$ for $y$ and applying axioms GC3 and A6.
It is not immediately clear that this equation is derivable from the
axioms of \ACPecr.
\begin{proposition}[Derivability Generalization Axiom R6]
\label{prop-R6}
The equation
$a \seqc (\cond{x}{\retro \phi}{y}) = \cond{a \seqc x}{\phi}{a \seqc y}$
$\axiom{(R6')}$ is derivable from the axioms of \ACPecr.
\end{proposition}
\begin{proof}
The proof is the same to the proof of the corresponding proposition in
the setting of \ACPcr\ given in~\cite{BM05a}.
\qed
\end{proof}

Because of the addition of the retrospection operator, we cannot use the
Boolean algebras $\CondAlg{\kappa}$ here.
The algebras $\CondrAlg{\kappa}$ that we use here can be characterized
as the free $\kappa$-complete algebras over $\ACond$ from the class of
algebras with interpretations for the constants and operators of Boolean
algebras and the retrospection operator that satisfy the axioms of
Boolean algebras (Table~\ref{axioms-ACPec}) and axioms R1--R5  from
Table~\ref{axioms-ACPecr}.
We do not make this fully precise, but give an explicit construction
of the algebras $\CondrAlg{\kappa}$ instead.
Important to bear in mind is that not only the atomic conditions, but
also the results of applying the operation associated with the
retrospection operator a finite number of times to atomic conditions,
should not satisfy any equations except those derivable from the axioms.

Let $\ACondr = \Union \set{\ACond \x \set{i} \where i \in \omega}$ and
define $\funct{\prev}{\ACondr}{\ACondr}$ by
$\prev(\tup{\eta,i}) = \tup{\eta,i+1}$.
For any infinite cardinal $\kappa$, let $\CondAlgp{\kappa}$ be the free
$\kappa$-complete Boolean algebra over $\ACondr$.
Then the function $\prev$ extends to a unique $\kappa$-complete
endomorphism $\prev^*$ of $\CondAlgp{\kappa}$.
This endomorphism is a unary operation on $\CondAlgp{\kappa}$ that
satisfies axioms R1--R5 from Table~\ref{axioms-ACPecr} and preserves
$\infjoin C'$ for every $C' \subseteq \CondAlgp{\kappa}$ of cardinality
less then $\kappa$.
The algebra $\CondrAlg{\kappa}$ is the expansion of $\CondAlgp{\kappa}$
obtained by associating the operation $\prev^*$ with the operator
$\retro$.
We write $\FinCondrAlg$ for $\CondrAlg{\aleph_0}$.

The structural operational semantics of \ACPecr\ is described by the
transition rules for \ACPec\ with the second and third transition rule
for parallel composition and the one transition rule for left merge
replaced by the transition rules given in
Table~\ref{sos-ACPec-adapt-retro},
\begin{table}[!t]
\caption{Transition rules adapted to retrospection}
\label{sos-ACPec-adapt-retro}
\begin{ruletbl}
\Rule
{\astep{x}{\gact{\phi}{a}}{x'}}
{\astep{x \parc y}{\gact{\phi}{a}}{x' \parc \rshiftz(y)}}
\quad
\Rule
{\astep{y}{\gact{\phi}{a}}{y'}}
{\astep{x \parc y}{\gact{\phi}{a}}{\rshiftz(x) \parc y'}}
\quad
\Rule
{\astep{x}{\gact{\phi}{a}}{x'}}
{\astep{x \leftm y}{\gact{\phi}{a}}{x' \parc \rshiftz(y)}}
\end{ruletbl}
\end{table}
and the additional transition rules for retrospection given in
Table~\ref{sos-ACPecr-rshift}.%
\begin{table}[!t]
\caption{Additional transition rules for retrospection}
\label{sos-ACPecr-rshift}
\begin{ruletbl}
\Rule
{\isterm{x}{\phi}}
{\isterm{\rshiftz(x)}{\rshift{0}(\phi)}}
\quad
\Rule
{\astep{x}{\gact{\phi}{a}}{x'}}
{\astep{\rshiftz(x)}{\gact{\rshift{0}(\phi)}{a}}{\rshift{1}(x')}}
\\
\Rule
{\isterm{x}{\phi}}
{\isterm{\rshift{n}(x)}{\rshift{n}(\phi)}}
\quad
\Rule
{\astep{x}{\gact{\phi}{a}}{x'}}
{\astep{\rshift{n}(x)}{\gact{\rshift{n}(\phi)}{a}}{\rshift{n+1}(x')}}
\end{ruletbl}
\end{table}
Of course, the conditions involved are now taken from $\FinCondrAlg$
instead of $\FinCondAlg$.

\section{Full Retrospective Splitting Bisimulation Models of \ACPecr}
\label{sect-full-bisim-ACPecr}

The construction of the full splitting bisimulation models of \ACPecr\
differs from the construction of the full splitting bisimulation models
of \ACPec\ in the conditions involved and in the notion of splitting
bisimulation used.
The conditions are now taken from $\CondrAlg{\kappa}$ instead of
$\CondAlg{\kappa}$.
Henceforth, we write $\CondrAlgm{\kappa}$ for
$\CondrAlg{\kappa} \diff \set{\bot}$.

Let $\kappa$ be an infinite cardinal.
Then a $\kappa$-\emph{conditional transition system with retrospection}
$T$ consists of the following:
\begin{iteml}
\item
a set $S$ of \emph{states};
\item
a set ${\step{\ell}} \subseteq S \x S$,
for each $\ell \in \CondrAlgm{\kappa} \x \Act$;
\item
a set ${\sterm{\alpha}} \subseteq S$,
for each $\alpha \in \CondrAlgm{\kappa}$;
\item
an \emph{initial state} $s^0 \in S$.
\end{iteml}

For conditional transition systems with retrospection, reachability,
connectedness and connected part are defined exactly as for conditional
transition systems.

Let $\kappa$ be an infinite cardinal.
Then $\CTSer_\kappa$ is the set of all connected $\kappa$-conditional
transition systems with retrospection
$T = \tup{S,{\stepf},{\stermf},s^0}$ such that
$S \subset \States{\kappa}$ and the branching degree of $T$ is less than
$\kappa$.

Isomorphism between conditional transition systems with retrospection is
defined exactly as for conditional transition systems.
Splitting bisimilarity has to be adapted to the setting with
retrospection.

Let $T_1 = \tup{S_1,{\stepf}_1,{\stermf}_1,s^0_1} \in \CTSer_\kappa$
and $T_2 = \tup{S_2,{\stepf}_2,{\stermf}_2,s^0_2} \in \CTSer_\kappa$
(for an infinite cardinal $\kappa$).
Then a \emph{retrospective splitting bisimulation} $B$ between $T_1$ and
$T_2$ is a ternary relation
$B \subseteq S_1 \x \CondrAlg{\kappa} \x S_2$ such that
$B(s^0_1,\top,s^0_2)$ and for all $s_1$, $\beta$, $s_2$ such that
$B(s_1,\beta,s_2)$:
\begin{iteml}
\item
if $\astepi{s_1}{\gact{\alpha}{a}}{s_1'}$, then there is a set
$CS_2' \subseteq \CondrAlgm{\kappa} \x S_2$ of cardinality less than
$\kappa$ such that $\alpha \meet \beta \beloweq \infjoin \dom(CS_2')$
and for all $\tup{\alpha',s_2'} \in CS_2'$,
$\astepii{s_2}{\gact{\alpha'}{a}}{s_2'}$ and
$B(s_1',\retro \alpha',s_2')$;
\item
if $\astepii{s_2}{\gact{\alpha}{a}}{s_2'}$, then there is a set
$CS_1' \subseteq \CondrAlgm{\kappa} \x S_1$ of cardinality less than
$\kappa$ such that $\alpha \meet \beta \beloweq \infjoin \dom(CS_1')$
and for all $\tup{\alpha',s_1'} \in CS_1'$,
$\astepi{s_1}{\gact{\alpha'}{a}}{s_1'}$ and
$B(s_1',\retro \alpha',s_2')$;
\item
if $\istermi{s_1}{\alpha}$, then there is a set
$C' \subseteq \CondrAlgm{\kappa}$ of cardinality less than $\kappa$ such
that $\alpha \meet \beta \beloweq \infjoin C'$ and for all
$\alpha' \in C'$, $\istermii{s_2}{\alpha'}$;
\item
if $\istermii{s_2}{\alpha}$, then there is a set
$C' \subseteq \CondrAlgm{\kappa}$ of cardinality less than $\kappa$ such
that $\alpha \meet \beta \beloweq \infjoin C'$ and for all
$\alpha' \in C'$, $\istermi{s_1}{\alpha'}$.
\end{iteml}
Two conditional transition systems with retrospection
$T_1,T_2 \in \CTSer_\kappa$ are \emph{retrospective splitting bisimilar},
written $T_1 \rsbisim T_2$, if there exists a retrospective splitting
bisimulation $B$ between $T_1$ and $T_2$.
Let $B$ be a retrospective splitting bisimulation between $T_1$ and
$T_2$.
Then we say that $B$ is a retrospective splitting bisimulation
\emph{witnessing} $T_1 \rsbisim T_2$.

It is straightforward to see that ${\rsbisim}$ is an equivalence on
$\CTSer_\kappa$.
Let $T \in \CTSer_\kappa$.
Then we write $\rsbeqvc{T}$ for
$\set{T' \in \CTSer_\kappa \where T \rsbisim T'}$, i.e.\ the
${\rsbisim}$\,-equivalence class of $T$.
We write $\CTSer_\kappa / {\rsbisim}$ for the set of equivalence classes
$\set{\rsbeqvc{T} \where T \in \CTSer_\kappa}$.

The elements of $\CTSer_\kappa$ and operations on $\CTSer_\kappa$ to be
associated with the constants and operators of \ACPec\ are defined
exactly as the elements of $\CTSe_\kappa$ and operations on $\CTSe_\kappa$
associated with them, except for $\parc$, $\leftm$ and $\commm$.
The operations on $\CTSer_\kappa$ that we associate with $\parc$,
$\leftm$, $\commm$, $\rshiftz$ and $\rshift{n}$ call for unfolding of
transition systems from $\CTSer_\kappa$.

For the sake of unfolding, it is assumed that, for each infinite
cardinal $\kappa$, $\States{\kappa}$ has the following closure
property:%
\footnote
{We write $\emptyseq$ for the empty sequence,
 $\seq{e}$ for the sequence having $e$ as sole element and
 $\sigma \conc \sigma'$ for the concatenation of sequences $\sigma$
 and $\sigma'$; and
 we use $\seq{e_1,\ldots,e_n}$ as a shorthand for
 $\seq{e_1} \conc \ldots \conc \seq{e_n}$.
}
\begin{ldispl}
\mathrm{for}\; \mathrm{all}\; S \subseteq \States{\kappa},\;
\set{\pi \conc \seq{s} \where
     \pi \in \seqof{(S \x (\CondrAlg{\kappa} \x \Act))}, s \in S}
 \subseteq
\States{\kappa}\;.
\end{ldispl}
We write $\paths'(S)$ for the set
$\set{\pi \conc \seq{s} \where
      \pi \in \seqof{(S \x (\CondrAlg{\kappa} \x \Act))}, s \in S}$.
The function $\funct{\nr}{\paths'(S)}{\Nat}$ is defined by
\begin{ldispl}
\begin{aeqns}
\nr(\seq{s}) & = & 0\;,
\\
\nr(\pi \conc \seq{s,\ell,s'}) & = & \nr(\pi \conc \seq{s}) + 1\;.
\end{aeqns}
\end{ldispl}

The elements of $\paths'(S)$, for an $S \subseteq \States{\kappa}$,
can be looked upon as potential paths of a $\kappa$-conditional
transition system with $S$ as set of states.
A path of a transition system
$\tup{S,{\stepf},{\stermf},s^0} \in \CTSer_\kappa$ is a finite
alternating sequence $\seq{s_0,\ell_1,s_1,\ldots,\ell_n,s_n}$ of states
from $S$ and labels from $\CondrAlg{\kappa} \x \Act$ such that
$s_0 =s^0$ and $\astep{s_i}{\ell_{i+1}}{s_{i+1}}$ for all $i < n$.
The state $s_n$ is called the state in which the path ends.

Let $T = \tup{S,{\stepf},{\stermf},s^0} \in \CTSer_\kappa$.
Then the set of \emph{paths} of $T$, written $\paths(T)$, is the
smallest subset of $\paths'(S)$ such that:
\begin{iteml}
\item
$\seq{s^0} \in \paths(T)$,
\item
if $\pi \conc \seq{s} \in \paths(T)$ and $\astep{s}{\ell}{s'}$,
then $\pi \conc \seq{s,\ell,s'} \in \paths(T)$.
\end{iteml}
In order to unfold a transition system, we need for each state $s$ of
the original transition system, for each different path that ends in
state $s$, a different state in the unfolded transition system.
The obvious choice is to take the paths concerned as states.

Let $T = \tup{S,{\stepf},{\stermf},s^0} \in \CTSer_\kappa$.
Then the \emph{unfolding} of $T$, written $\unf(T)$, is defined as
follows:
\begin{ldispl}
\unf(T) = \tup{S',{\stepf}',{\stermf}',s^0{}'}\;,
\end{ldispl}
where
\begin{ldispl}
\begin{aeqns}
S' & = & \paths(T)\;,
\end{aeqns}
\end{ldispl}
and for every
$\ell \in \CondrAlgm{\kappa} \x \Act$ and
$\alpha \in  \CondrAlgm{\kappa}$:
\begin{ldispl}
\begin{aeqns}
{\stepp{\ell}} & = &
\set{\tup{\pi \conc \seq{s}, \pi \conc \seq{s,\ell,s'}} \where
     \pi \conc \seq{s} \in \paths(T), \astep{s}{\ell}{s'}}\;,
\\
{\stermp{\alpha}} & = &
\set{\pi \conc \seq{s} \where
     \pi \conc \seq{s} \in \paths(T), \isterm{s}{\alpha}}\;,
\\
s^0{}' & = & \seq{s^0}\;.
\end{aeqns}
\end{ldispl}

The functions $\updl$ and $\updr$ defined next will be used in the
definition of parallel composition on $\CTSer_\kappa$ to adapt the
retrospection in steps originating from the first operand and the
second operand, respectively.

Let $S_1,S_2 \subseteq \States{\kappa}$.
Then the functions
$\funct{\upd_i}
 {\CondrAlgm{\kappa} \x \paths'(S_1 \x S_2)}{\CondrAlgm{\kappa}}$,
for $i = 1,2$, are defined by
\begin{ldispl}
\begin{aeqns}
\upd_i(\alpha,\seq{\tup{s_1,s_2}}) & = & \alpha\;,
\\
\upd_i(\alpha,\seq{\tup{s_1,s_2},\ell,\tup{s_1',s_2'}} \conc \pi') & = &
\upd_i(\alpha,\seq{\tup{s_1',s_2'}} \conc \pi') &
\;\mathrm{if}\; s_i \neq s_i'\;,
\\
\upd_i(\alpha,\seq{\tup{s_1,s_2},\ell,\tup{s_1',s_2'}} \conc \pi') & = &
\\
\multicolumn{3}{r}
{\upd_i(\rshift{\nr_i(\seq{\tup{s_1',s_2'}} \conc \pi')}(\alpha),
       \seq{\tup{s_1',s_2'}} \conc \pi')} &
\;\mathrm{if}\; s_i = s_i'\;.
\end{aeqns}
\end{ldispl}
where
\begin{ldispl}
\begin{aeqns}
\nr_i(\seq{\tup{s_1,s_2}}) & = & 0\;,
\\
\nr_i(\seq{\tup{s_1,s_2},\ell,\tup{s_1',s_2'}} \conc \pi') & = &
\nr_i(\seq{\tup{s_1',s_2'}} \conc \pi') + 1 &
\;\mathrm{if}\; s_i \neq s_i'\;,
\\
\nr_i(\seq{\tup{s_1,s_2},\ell,\tup{s_1',s_2'}} \conc \pi') & = &
\nr_i(\seq{\tup{s_1',s_2'}} \conc \pi') &
\;\mathrm{if}\; s_i = s_i'\;.
\end{aeqns}
\end{ldispl}
Henceforth, we write $\upd(\alpha_1,\alpha_2,\pi)$ for
$\updl(\alpha_1,\pi) \meet \updr(\alpha_2,\pi)$.

We proceed with associating operations on $\CTSer_\kappa$ with the
operators $\parc$, $\leftm$, $\commm$, $\rshiftz$ and $\rshift{n}$.

We associate with the additional operator $\parc$ an operation
$\tsparcr$ on $\CTSer_\kappa$ as follows.
\begin{iteml}
\item
Let $T_1, T_2 \in \CTSer_\kappa$.
Suppose that
$\unf(T_i) = \tup{S_i,{\stepf}_i,{\stermf}_i,s^0_i}$ for $i = 1,2$, and
$\unf(\unf(T_1) \tsparc \unf(T_2)) = \tup{S,{\stepf},{\stermf},s^0}$.
Then
\begin{ldispl}
T_1 \tsparcr T_2 = \tup{S,{\stepf}',{\stermf}',s^0}\;,
\end{ldispl}
where for every $\tup{\alpha,a} \in \CondrAlgm{\kappa} \x \Act$ and
$\alpha'' \in \CondrAlgm{\kappa}$:
\pagebreak[2]
\begin{ldispl}
\begin{aeqns}
{\step{\tup{\alpha,a}}}{}' & = &
\set{\tup{\pi \conc \seq{\tup{s_1,s_2}},
          \pi' \conc \seq{\tup{s_1',s_2'}}} \where
     s_1 \neq s_1' \And s_2 = s_2' \And {}
\\ & & \phantom{\{\,}
     \smash{\OR{\alpha' \in \CondrAlgm{\kappa}}}
      (\astep{\pi \conc \seq{\tup{s_1,s_2}}}
             {\gact{\alpha'}{a}}{\pi' \conc \seq{\tup{s_1',s_2'}}}
        \And {}
\\ & & \hfill
       \updl(\alpha',\pi \conc \seq{\tup{s_1,s_2}}) = \alpha)}%
\phantom{\;,}
\\ & \union &
\set{\tup{\pi \conc \seq{\tup{s_1,s_2}},
          \pi' \conc \seq{\tup{s_1',s_2'}}} \where
     s_1 = s_1' \And s_2 \neq s_2' \And {}
\\ & & \phantom{\{\,}
     \smash{\OR{\alpha' \in \CondrAlgm{\kappa}}}
      (\astep{\pi \conc \seq{\tup{s_1,s_2}}}
             {\gact{\alpha'}{a}}{\pi' \conc \seq{\tup{s_1',s_2'}}}
        \And {}
\\ & & \hfill
       \updr(\alpha',\pi \conc \seq{\tup{s_1,s_2}}) = \alpha)}%
\phantom{\;,}
\phantom{{\step{\tup{\alpha,a}}}{}'}
\\ & \union &
\set{\tup{\pi \conc \seq{\tup{s_1,s_2}},
          \pi' \conc \seq{\tup{s_1',s_2'}}} \where
\\ & & \phantom{\{\,}
     \smash{\OR{\alpha',\beta' \in \CondrAlgm{\kappa}, a',b' \in \Act}}
      (\astep{\pi \conc \seq{\tup{s_1,s_2}}}
             {\gact{\alpha' \meet \beta'}{a}}
             {\pi' \conc \seq{\tup{s_1',s_2'}}}
        \And {}\phantom{)\;.}
\\ & & \phantom{\{\,}
\phantom
{\smash{\OR{\alpha',\beta' \in \CondrAlgm{\kappa}, a',b' \in \Act}}(}
       \astepi{s_1}{\gact{\alpha'}{a'}}{s_1'} \And
       \astepii{s_2}{\gact{\beta'}{b'}}{s_2'} \And {}
\\ & & \phantom{\{\,}
\phantom
{\smash{\OR{\alpha',\beta' \in \CondrAlgm{\kappa}, a',b' \in \Act}}(}
       \upd(\alpha',\beta',\pi \conc \seq{\tup{s_1,s_2}}) = \alpha
        \And {}
\\ & & \hfill
       a' \commm b' = a)}\;,
\eqnsep
{\sterm{\alpha''}}{}' & = &
\set{\pi \conc \seq{\tup{s_1,s_2}} \where
\\ & & \phantom{\{\,}
     \smash{\OR{\alpha',\beta' \in \CondrAlgm{\kappa}}}
      (\isterm{\pi \conc \seq{\tup{s_1,s_2}}}{\alpha' \meet \beta'} \And
       \istermi{s_1}{\alpha'} \And
       \istermii{s_2}{\beta'} \And {}
\\ & & \hfill
       \upd
        (\alpha',\beta',\pi \conc \seq{\tup{s_1,s_2}}) = \alpha'')}\;.
\end{aeqns}
\end{ldispl}
\end{iteml}
\begin{remark}
The operation $\tsparcr$ on $\CTSer_\kappa$ is defined above in a
step-by-step way.
The basic idea behind this definition is twofold:
\begin{iteml}
\item
$T_1 \tsparcr T_2$ can be obtained by first composing $T_1$ and $T_2$
to $T_1 \tsparc T_2$ and then adapting the retrospections in steps of
$T_1 \tsparc T_2$;
\item
unfolding of $T_1 \tsparc T_2$ is needed before the actual adaptations
can take place because the adaptation of the retrospection in a step
may be different for the different paths that end in the state from
which the step starts.
\end{iteml}
Somewhat surprisingly, in addition, $T_1$ and $T_2$ must be unfolded
before the actual composition takes place.
In a step where an action of $T_1$ and an action of $T_2$ are performed
synchronously, the condition under which the action of $T_1$ can be
performed and the condition under which the action of $T_2$ can be
performed are needed to adapt the retrospection in that step correctly.
If $T_1$ and $T_2$ are not unfolded before the actual composition takes
place, in general, those conditions cannot be determined uniquely.
\end{remark}
The operations on $\CTSer_\kappa$ to be associated with the additional
operators $\leftm$ and $\commm$ are defined analogously.
The operations on $\CTSer_\kappa$ to be associated with the additional
operators $\encap{H}$ are defined exactly as the operations on
$\CTSe_\kappa$ associated with them.
We associate with the additional operators $\rshift{n}$ operations
$\tsrshiftr{n}$ on $\CTSer_\kappa$ as follows.
\begin{iteml}
\item
Let $T \in \CTSer_\kappa$.
Suppose that
$\unf(T) = \tup{S,{\stepf},{\stermf},s^0}$.
Then
\begin{ldispl}
\tsrshiftr{n}(T) = \tup{S,{\stepf}',{\stermf}',s^0}\;,
\end{ldispl}
where for every $\tup{\alpha,a} \in \CondrAlgm{\kappa} \x \Act$ and
$\alpha'' \in \CondrAlgm{\kappa}$:
\begin{ldispl}
\begin{aeqns}
{\step{\tup{\alpha,a}}}{}' & = &
\set{\tup{\pi \conc \seq{s},\pi' \conc \seq{s'}} \where
\\ & & \phantom{\{\,}
     \OR{\alpha' \in \CondrAlgm{\kappa}}
      (\astep{\pi \conc \seq{s}}{\gact{\alpha'}{a}}{\pi' \conc \seq{s'}}
        \And
       \rshift{\nr(\pi) + n}(\alpha') = \alpha)}\;,
\eqnsep
{\sterm{\alpha''}}{}' & = &
\set{\pi \conc \seq{s} \where
     \OR{\alpha' \in \CondrAlgm{\kappa}}
      (\isterm{\pi \conc \seq{s}}{\alpha'} \And
       \rshift{\nr(\pi) + n}(\alpha') = \alpha'')}\;.
\end{aeqns}
\end{ldispl}
\end{iteml}
The operation on $\CTSer_\kappa$ to be associated with the additional
operator $\rshiftz$ is the same as the operation on $\CTSer_\kappa$
associated with $\rshift{0}$.

We can show that retrospective splitting bisimilarity is a congruence
with respect to the operations on $\CTSer_\kappa$ associated with the
operators of \ACPecr.
\begin{proposition}[Congruence]
\label{prop-congruence-ACPecr}
Let $\kappa$ be an infinite cardinal.
Then for all $T_1,T_2,T_1',T_2' \in \CTSer_\kappa$ and
$\alpha \in \CondAlg{\kappa}$,
$T_1 \rsbisim T_1'$ and $T_2 \rsbisim T_2'$ imply
$T_1 \tsaltcr T_2 \rsbisim T_1' \tsaltcr T_2'$,
$T_1 \tsseqcr T_2 \rsbisim T_1' \tsseqcr T_2'$,
$\alpha \tsgcr T_1 \rsbisim \alpha \tsgcr T_1'$,
$T_1 \tsparcr T_2 \rsbisim T_1' \tsparcr T_2'$,
$T_1 \tsleftmr T_2 \rsbisim T_1' \tsleftmr T_2'$,
$T_1 \tscommmr T_2 \rsbisim T_1' \tscommmr T_2'$,
$\tsencapr{H}(T_1) \rsbisim \tsencapr{H}(T_1')$,
$\tsrshiftzr(T_1) \rsbisim \tsrshiftzr(T_1')$ and
$\tsrshiftr{n}(T_1) \rsbisim \tsrshiftr{n}(T_1')$.
\end{proposition}
\begin{proof}
For all operations, witnessing splitting bisimulations are constructed
in the same way as in the congruence proofs for the corresponding
operations on $\CTSr_\kappa$ given in~\cite{BM05a}.
\qed
\end{proof}

The ingredients of the
\emph{full retrospective splitting bisimulation models} $\CPrcer_\kappa$
of \ACPecr, one for each infinite cardinal $\kappa$, are defined as
follows:
\begin{ldispl}
\begin{aeqns}
\cP & = & \CTSer_\kappa / {\rsbisim}\;,
\eqnsep
\sdeadr & = & \rsbeqvc{\tsdeadr}\;,
\eqnsep
\sepr & = & \rsbeqvc{\tsepr}\;,
\eqnsep
\sactr{a} & = & \rsbeqvc{\tsactr{a}}\;,
\eqnsep
\rsbeqvc{T_1} \saltcr \rsbeqvc{T_2} & = & \rsbeqvc{T_1 \tsaltcr T_2}\;,
\eqnsep
\rsbeqvc{T_1} \sseqcr \rsbeqvc{T_2} & = & \rsbeqvc{T_1 \tsseqcr T_2}\;,
\eqnsep
\alpha \sgcr \rsbeqvc{T_1} & = & \rsbeqvc{\alpha \tsgcr T_1}\;,
\end{aeqns}
\quad\;\;
\begin{aeqns}
{} \eqnsep
\rsbeqvc{T_1} \sparcr \rsbeqvc{T_2} & = & \rsbeqvc{T_1 \tsparcr T_2}\;,
\eqnsep
\rsbeqvc{T_1} \sleftmr \rsbeqvc{T_2} & = & \rsbeqvc{T_1 \tsleftmr T_2}\;,
\eqnsep
\rsbeqvc{T_1} \scommmr \rsbeqvc{T_2} & = & \rsbeqvc{T_1 \tscommmr T_2}\;,
\eqnsep
\sencapr{H}(\rsbeqvc{T_1}) & = & \rsbeqvc{\tsencapr{H}(T_1)}\;,
\eqnsep
\srshiftzr(\rsbeqvc{T_1}) & = & \rsbeqvc{\tsrshiftzr(T_1)}\;,
\eqnsep
\srshiftr{n}(\rsbeqvc{T_1}) & = & \rsbeqvc{\tsrshiftr{n}(T_1)}\;.
\end{aeqns}
\end{ldispl}
The operations on $\CTSer_\kappa / {\rsbisim}$ are well-defined because
${\rsbisim}$ is a congruence with respect to the corresponding
operations on $\CTSer_\kappa$.

The structures $\CPrcer_\kappa$ are models of \ACPecr.
\begin{theorem}[Soundness of \ACPecr]
\label{theorem-soundness-ACPecr}
For each infinite cardinal $\kappa$, we have
$\Sat{\CPrcer_\kappa}{\ACPecr}$.
\end{theorem}
\begin{proof}
Because ${\rsbisim}$ is a congruence, it is sufficient to show that all
axioms are sound.
The soundness of all axioms follows straightforwardly from the
definition of $\CPrcer_\kappa$.
\qed
\end{proof}
For all axioms that are in common with \ACPcr, the proof of soundness
with respect to $\CPrcer_\kappa$ follows the same line as the proof of
soundness with respect to $\CPrcr_\kappa$.

In the full retrospective splitting bisimulation models of \ACPecr,
guarded recursive specifications over \ACPecr\ have unique solutions.
\begin{theorem}[Unique solutions in $\CPrcer_\kappa$]
\label{theorem-uniqueness-retro}
\sloppy
For each infinite cardinal $\kappa$, guarded recursive specifications
over \ACPecr\ have unique solutions in $\CPrcer_\kappa$.
\end{theorem}
\begin{proof}
The proof is analogous to the proof of the corresponding property for
the full retrospective splitting bisimulation models of \ACPcr\ given
in~\cite{BM05a}.
\qed
\end{proof}
Thus, the full retrospective splitting bisimulation models
$\CPrcer_\kappa{}'$ of \ACPecr\ with guarded recursion are simply the
expansions of the full retrospective splitting bisimulation models
$\CPrcer_\kappa$ of \ACPecr\ obtained by associating with each constant
$\rec{X}{E}$ the unique solution of $E$ for $X$ in the full
retrospective splitting bisimulation model concerned.

\section{Evaluation of Retrospective Conditions}
\label{sect-cond-eval-retro}

In this section, we add condition evaluation operators and generalized
condition evaluation operators to \ACPecr.
As in the case of \ACPec, these operators require to fix an infinite
cardinal $\lambda$.
By doing so, full retrospective splitting bisimulation models with
domain $\CTSer_{\kappa} / {\rsbisim}$ for $\kappa > \lambda$ are
excluded.

Henceforth, we write $\CondrHom{\lambda}$ for the set of all
$\lambda$-complete endomorphisms of $\CondrAlg{\lambda}$.

In the case of \ACPecr, there are $\lambda$-complete condition
evaluation operators $\funct{\ceval{h}}{\Proc}{\Proc}$ and
$\funct{\ceval{h}}{\Cond}{\Cond}$, and generalized $\lambda$-complete
condition evaluation operators $\funct{\gceval{h}}{\Proc}{\Proc}$ and
$\funct{\gceval{h}}{\Cond}{\Cond}$, for each $h \in \CondrHom{\lambda}$.
We also need the following auxiliary operators:
\begin{iteml}
\item
for each $h \in \CondrHom{\lambda}$, $n \in \Nat$,\,
the unary \emph{retrospection update} operator
$\funct{\rupd{h}{n}}{\Proc}{\Proc}$;
\item
for each $h \in \CondrHom{\lambda}$, $n \in \Nat$,\,
the unary \emph{retrospection update} operator
$\funct{\rupd{h}{n}}{\Cond}{\Cond}$.
\end{iteml}

In the case of \ACPecr, it is assumed that a fixed but arbitrary function
$\funct{\eff}{\Act \x \CondrHom{\lambda}}{\CondrHom{\lambda}}$ has been
given.
The function $\eff$ is extended to $\Actd$ such that $\eff(\dead,h) = h$
for all $h \in \CondrHom{\lambda}$.

The condition evaluation operators and generalized condition evaluation
operators cannot be added to \ACPecr\ in the same way as they are added
to \ACPec.
First of all, retrospective conditions may refer back too far to be
evaluated.
The effect is that, in condition evaluation or generalized condition
evaluation of a process according to some endomorphism, the
retrospective conditions that refer back further than the beginning of
the process have to be left unevaluated.
This is accomplished by the retrospection update operators mentioned
above.
In the case of generalized condition evaluation, there is another
complication.
Recall that generalized condition evaluation allows the results of
condition evaluation to change by performing an action.
In the presence of retrospection, different parts of a condition may
have to be evaluated differently because of such changes.
The effect is that, in generalized condition evaluation of a process
according to some endomorphism, after an action of the process is
performed, the subsequent retrospective conditions that refer back to
the beginning of the process have to be evaluated according to that
endomorphism as well.
This is also accomplished by the retrospection update operators
mentioned above.

In the case of \ACPecr, the additional axioms for $\ceval{h}$ and
$\gceval{h}$, where $h \in \CondrHom{\lambda}$, are the axioms given in
Tables~\ref{axioms-ceval-gceval-retro} and~\ref{axioms-retro-upd}.%
\begin{table}[!t]
\lcaption{New axioms for (generalized) condition evaluation}
  {$a \in \Actd$}
\label{axioms-ceval-gceval-retro}
\begin{eqntbl}
\begin{axcol}
\ceval{h}(\ep) = \ep                                   & \axiom{CE1T} \\
\ceval{h}(a \seqc x) =
                    a \seqc \ceval{h}(\rupd{h}{1}(x))  & \axiom{CE2R} \\
\ceval{h}(x \altc y) = \ceval{h}(x) \altc \ceval{h}(y) & \axiom{CE3}  \\
\ceval{h}(\phi \gc x) =
                    \rupd{h}{0}(\phi) \gc \ceval{h}(x) & \axiom{CE4R}
\eqnsep
\gceval{h}(\ep) = \ep                                  & \axiom{GCE1T}\\
\gceval{h}(a \seqc x) =
            a \seqc \gceval{\eff(a,h)}(\rupd{h}{1}(x)) & \axiom{GCE2R}\\
\gceval{h}(x \altc y) =
                     \gceval{h}(x) \altc \gceval{h}(y) & \axiom{GCE3} \\
\gceval{h}(\phi \gc x) =
                   \rupd{h}{0}(\phi) \gc \gceval{h}(x) & \axiom{GCE4R}
\end{axcol}
\end{eqntbl}
\end{table}
\begin{table}[!t]
\lcaption{Axioms for retrospection update}
  {$a \in \Actd$, $\eta \in \ACond$,
   $\eta' \in \ACond \union \set{\bot,\top}$}
\label{axioms-retro-upd}
\begin{eqntbl}
\begin{axcol}
\rupd{h}{n}(\ep) = \ep                                  & \axiom{RU1T}\\
\rupd{h}{n}(a \seqc x) = a \seqc \rupd{h}{n+1}(x)       & \axiom{RU2} \\
\rupd{h}{n}(x \altc y) =
                    \rupd{h}{n}(x) \altc \rupd{h}{n}(y) & \axiom{RU3} \\
\rupd{h}{n}(\phi \gc x) =
                   \rupd{h}{n}(\phi) \gc \rupd{h}{n}(x) & \axiom{RU4}
\\ {} \\
\rupd{h}{n}(\bot) = \bot                                & \axiom{RU5} \\
\rupd{h}{n}(\top) = \top                                & \axiom{RU6}
\end{axcol}
\quad\quad
\begin{axcol}
\rupd{h}{0}(\eta) = \eta'   \hfill \mif h(\eta) = \eta' & \axiom{RU7} \\
\rupd{h}{n+1}(\eta) = \eta                              & \axiom{RU8} \\
\rupd{h}{n}(\bcompl \phi) = \bcompl \rupd{h}{n}(\phi)   & \axiom{RU9} \\
\rupd{h}{n}(\phi \join \psi) =
              \rupd{h}{n}(\phi) \join \rupd{h}{n}(\psi) & \axiom{RU10}\\
\rupd{h}{n}(\phi \meet \psi) =
              \rupd{h}{n}(\phi) \meet \rupd{h}{n}(\psi) & \axiom{RU11}\\
\rupd{h}{0}(\retro \phi) = \retro \phi                  & \axiom{RU12}\\
\rupd{h}{n+1}(\retro \phi) =
                             \retro \rupd{h}{n}(\phi)   & \axiom{RU13}
\end{axcol}
\end{eqntbl}
\end{table}
These additional axioms differ from the additional axioms in the absence
of retrospection (Tables~\ref{axioms-ceval} and~\ref{axioms-gceval}) in
that axioms CE2, CE4, GCE2 and GCE4 have been replaced by axioms CE2R,
CE4R, GCE2R and GCE4R, and axioms CE6--CE11 by axioms RU1T and
RU2--RU13.
Axioms CE2R and CE4R, together with axioms RU1T and RU2--RU13, show
that, in condition evaluation of a process, retrospective conditions
that refer back further than the beginning of the process are not at all
evaluated.
Similarly, axioms GCE2R and GCE4R, together with axioms RU1T and
RU2--RU13, show that, in generalized condition evaluation of a process,
retrospective conditions that refer back further than the beginning of
the process are not at all evaluated.
Moreover, axiom GCE2R, together with axioms RU1T and RU2--RU13, shows
that, in generalized condition evaluation of a process according to some
endomorphism, after an action of the process is performed, the
subsequent retrospective conditions that refer back to the beginning of
the process are evaluated according to that endomorphism as well.

The structural operational semantics of \ACPecr\ extended with condition
evaluation and generalized condition evaluation is described by the
transition rules for \ACPecr\ and the transition rules given in
Table~\ref{sos-ceval-gceval-retro}.%
\begin{table}[!t]
\caption{New transition rules for (generalized) condition evaluation}
\label{sos-ceval-gceval-retro}
\begin{ruletbl}
\RuleC
{\isterm{x}{\phi}}
{\isterm{\ceval{h}(x)}{\rupd{h}{0}(\phi)}}
{\rupd{h}{0}(\phi) \neq \bot}
\quad
\RuleC
{\astep{x}{\gact{\phi}{a}}{x'}}
{\astep{\ceval{h}(x)}{\gact{\rupd{h}{0}(\phi)}{a}}
       {\ceval{h}(\rupd{h}{1}(x'))}}
{\rupd{h}{0}(\phi) \neq \bot}
\\
\RuleC
{\isterm{x}{\phi}}
{\isterm{\gceval{h}(x)}{\rupd{h}{0}(\phi)}}
{\rupd{h}{0}(\phi) \neq \bot}
\quad
\RuleC
{\astep{x}{\gact{\phi}{a}}{x'}}
{\astep{\gceval{h}(x)}{\gact{\rupd{h}{0}(\phi)}{a}}
       {\gceval{\eff(a,h)}(\rupd{h}{1}(x'))}}
{\rupd{h}{0}(\phi) \neq \bot}
\\
\RuleC
{\isterm{x}{\phi}}
{\isterm{\rupd{h}{n}(x)}{\rupd{h}{n}(\phi)}}
{\rupd{h}{n}(\phi) \neq \bot}
\quad
\RuleC
{\astep{x}{\gact{\phi}{a}}{x'}}
{\astep{\rupd{h}{n}(x)}{\gact{\rupd{h}{n}(\phi)}{a}}{\rupd{h}{n+1}(x')}}
{\rupd{h}{n}(\phi) \neq \bot}
\end{ruletbl}
\end{table}

The full retrospective splitting bisimulation models of \ACPecr\ with
condition evaluation and/or generalized condition evaluation are not
simply the expansions of the full retrospective splitting bisimulation
models $\CPrcer_\kappa$ of \ACPecr, for infinite cardinals
$\kappa \leq \lambda$, obtained by associating with each operator
$\ceval{h}$ and/or $\gceval{h}$ the corresponding re-labeling operation
on conditional transition systems with retrospection.
As suggested by the structural operational semantics of \ACPecr\ extended
with condition evaluation and generalized condition evaluation, these
re-labeling operations have to be adapted in a way similar to the way in
which parallel composition had to be adapted to the case with
retrospection in Section~\ref{sect-full-bisim-ACPecr}.
As mentioned before, full retrospective splitting bisimulation models
with domain $\CTSer_{\kappa} / {\rsbisim}$ for $\kappa > \lambda$ are
excluded.

Proposition~\ref{prop-generalization-gceval}, stating that the
generalized $\lambda$-complete condition evaluation operators supersede
the $\lambda$-complete condition evaluation operators in the setting of
\ACPec, goes through in the setting of \ACPecr.

Adding state operators to \ACPecr\ can be done on the same lines as
adding generalized evaluation operators to \ACPecr, but is more
complicated.
Roughly speaking, signal emission can be added to \ACPecr\ in the same
way as it is added to \ACPec\ provided that signals are taken from
$\FinCondAlg$.
No adaptations like for generalized condition evaluation are needed
because signal emission corresponds to condition evaluation that does
not persist over performing an action.
This property also points at one of the differences between the
signal-emission approach to condition evaluation and the other
approaches treated in this paper: retrospection has to be resolved in
the signal-emission approach before condition evaluation can take place.
The case where signals are taken from $\FinCondrAlg$ is expected to be
too complicated to handle.

\section{An Application of \ACPecr}
\label{sect-appl-retro}

The ultimate applications of a process algebra that includes conditional
expressions of some form are the ones that remain entirely within the
domain of process algebra.
Such applications are by their nature extensions as well.
We outline one interesting application of this kind in the setting of
\ACPecr.

We take the set $\set{\just{a} \where a \in \Act}$ of
\emph{last action conditions} as the set of atomic conditions $\ACond$.
The intuition is that $\just{a}$ indicates that action $a$ is performed
just now.
The retrospection operator now allows for using conditions which express
that a certain number of steps ago a certain action must have been
performed.

Because we remain entirely within the domain of process algebra
some additional axioms are needed.
They are given in Table~\ref{axioms-ACPecr-just}.%
\begin{table}[!t]
\lcaption{Additional axioms for last action conditions}{$a \in \Act$}
\label{axioms-ACPecr-just}
\begin{eqntbl}
\begin{axcol}
a \seqc x = a \seqc (\just{a} \gc x)                   & \axiom{J}
\end{axcol}
\end{eqntbl}
\end{table}
Moreover, axioms CM7 (Table~\ref{axioms-ACPec}) and RS7
(Table~\ref{axioms-ACPecr}) must be replaced by axioms CM7J and
RS7Ja--RS7Jb from Table~\ref{axioms-ACPecr-adapt-just}.
\begin{table}[!t]
\lcaption{Axioms adapted to last action conditions}
  {$a,b \in \Actd$, $c \in \Act$}
\label{axioms-ACPecr-adapt-just}
\begin{eqntbl}
\begin{axcol}
a \seqc x \commm b \seqc y =
(a \commm b) \seqc (\laupd{a}{0}(x) \parc \laupd{b}{0}(y))
                                                       & \axiom{CM7J}
\\ {} \\
\rshift{0}(\just{c}) = \retro \just{c}                 & \axiom{RS7Ja}\\
\rshift{n+1}(\just{c}) = \just{c}                      & \axiom{RS7Jb}
\\ {} \\ {} \\
\laupd{a}{n}(\ep) = \ep                                & \axiom{LAU1T}\\
\laupd{a}{n}(b \seqc x) = b \seqc \laupd{a}{n+1}(x)    & \axiom{LAU2}\\
\laupd{a}{n}(x \altc y) =
                 \laupd{a}{n}(x) \altc \laupd{a}{n}(y) & \axiom{LAU3}\\
\laupd{a}{n}(\phi \gc x) =
                \laupd{a}{n}(\phi) \gc \laupd{a}{n}(x) & \axiom{LAU4}
\end{axcol}
\quad
\begin{axcol}
\laupd{a}{n}(\bot) = \bot                              & \axiom{LAU5}\\
\laupd{a}{n}(\top) = \top                              & \axiom{LAU6}\\
\laupd{a}{0}(\just{c}) = \bot     \hfill \mif a \neq c & \axiom{LAU7}\\
\laupd{a}{0}(\just{c}) = \top     \hfill \mif a = c    & \axiom{LAU8}\\
\laupd{a}{n+1}(\just{c}) = \just{c}                    & \axiom{LAU9}\\
\laupd{a}{n}(\bcompl \phi) =
                            \bcompl \laupd{a}{n}(\phi) & \axiom{LAU10}\\
\laupd{a}{n}(\phi \join \psi) =
           \laupd{a}{n}(\phi) \join \laupd{a}{n}(\psi) & \axiom{LAU11}\\
\laupd{a}{n}(\phi \meet \psi) =
           \laupd{a}{n}(\phi) \meet \laupd{a}{n}(\psi) & \axiom{LAU12}\\
\laupd{a}{0}(\retro \phi) = \retro \phi                & \axiom{LAU13}\\
\laupd{a}{n+1}(\retro \phi) =
                             \retro \laupd{a}{n}(\phi) & \axiom{LAU14}
\end{axcol}
\end{eqntbl}
\end{table}
Axiom CM7 must be replaced by axiom CM7J because, after performing
$a \commm b$, it makes no sense to refer back to the actions performed
just now by the processes originally following $a$ and $b$ in the
process following $a \commm b$.
Retrospective conditions in the process originally following $a$ that
indicate that $a$ is performed just now should be evaluated to $\top$
and the ones that indicate that another action is performed just now
should be evaluated to $\bot$.
Retrospective conditions in the process originally following $b$ should
be evaluated analogously.
This is accomplished by the auxiliary operators
$\funct{\laupd{a}{n}}{\Proc}{\Proc}$ and
$\funct{\laupd{a}{n}}{\Cond}{\Cond}$
(for each $a \in \Actd$ and $n \in \Nat$) of which the defining axioms
are LAU1T and LAU2--LAU14 from Table~\ref{axioms-ACPecr-adapt-just}.
Axiom RS7 must be replaced by axioms RS7Ja and RS7Jb because of the
retrospective nature of last action conditions.
We mean by this that $\just{a}$ can be viewed as a condition of the form
$\retro \eta$, where $\eta$ indicates that action $a$ is performed next.
We have not introduced corresponding atomic conditions because their
use without restrictions would be problematic in alternative
composition.

From the axioms of \BPAdecr\ and the additional axiom J, we can
derive the equation
$a \seqc x \altc b \seqc y =
(a \altc b) \seqc (\just{a} \gc x \altc \just{b} \gc y)$.
It can be used to reduce the number of subprocesses of a
process.
For example, 
$a \seqc (a_1 \seqc a_1' \altc a_2 \seqc a_2') \altc
 b \seqc (b_1 \seqc b_1' \altc b_2 \seqc b_2') =
(a \altc b) \seqc
(\just{a} \gc
 (a_1 \altc a_2) \seqc (\just{a_1} \gc a_1' \altc \just{a_2} \gc a_2')
  \altc
 \just{b} \gc
 (b_1 \altc b_2) \seqc (\just{b_1} \gc b_1' \altc \just{b_2} \gc b_2'))$
shows a reduction from 7 subprocesses to 4 subprocesses.

In order to obtain the full retrospective splitting bisimulation models
of the extension of \ACPecr\ with last action conditions, retrospective
splitting bisimilarity has to be adapted:
in the definition of retrospective splitting bisimulation
(see Section~\ref{sect-full-bisim-ACPecr}), the two occurrences of
$B(s_1',\retro \alpha',s_2')$ must be replaced by
$B(s_1',\retro \alpha' \meet \just{a},s_2')$.

The operators $\laupd{a}{n}$ are reminiscent of the operators
$\rupd{h}{n}$.
In fact, if we would exclude full retrospective splitting bisimulation
models with domain $\CTSer_{\kappa} / {\rsbisim}$ for $\kappa$ greater
than some infinite cardinal $\lambda$,\, $\laupd{a}{n}$ could have been
replaced by $\rupd{{h_a}}{n}$, where $h_a \in \CondrHom{\lambda}$ for
$a \in \Act$ is defined by $h_a(\just{a}) = \top$ and
$h_a(\just{b}) = \bot$ if $a \neq b$ and
$h_\dead \in \CondrHom{\lambda}$ is defined by
$h_\dead(\just{a}) = \bot$.

We conclude with an example of the use of the retrospection operator
together with last action conditions.
\begin{example}
\label{example-service}
The example concerns a service that resembles the services considered
in~\cite{BM04d,BM05c}.
For any command $m$ from some set $M$, the service can be requested to
process command $m$ and it can be requested to report back what the
reply would be to the request to process command $m$.
We suppose that the service can be described by a function
$\funct{F}{\neseqof{M}}{\set{\True,\False,\Blocked}}$ with the property
that $F(\alpha) = \Blocked \Then F(\alpha \conc \seq{m}) = \Blocked$.
This function is called the \emph{reply} function of the service.
Given a reply function $F$ and a command $m$, the derived reply function
of $F$ after processing $m$, written $\derive{m}F$, is defined by
$\derive{m}F(\alpha) = F(\seq{m} \conc \alpha)$.
The connection between a reply function $F$ and the service described
by it can be understood as follows:
\begin{itemize}
\item
if $F(\seq{m}) \neq \Blocked$, the request to process command $m$ is
accepted by the service, the reply is $F(\seq{m})$ and the service
proceeds as described by $\derive{m}F$;
\item
if $F(\seq{m}) = \Blocked$, the request to process command $m$ is not
accepted by the service, the reply is $F(\seq{m})$ and the service
proceeds as described by $F$;
\item
the request to report back what the reply would be to the request to
process command $m$ is always accepted by the service, the reply is
$F(\seq{m})$ and the service proceeds as described by $F$.
\end{itemize}
Hence, the service can be viewed as the process defined by the guarded
recursive specification that consists of an equation
\begin{ldispl}
P_{G} =
\Altc{m \in M}
  (r(m) \altc r(?m)) \seqc s(G(\seq{m})) \seqc
   (\cond{P_{\derive{m} G}}
         {\retro \just{r(m)} \meet
          \bcompl{\just{s(\Blocked)}}}
         {P_{G}})
\end{ldispl}
for each reply function $G$.
Here, we write $r(m)$ for the action of receiving a request to process
command $m$, $r(?m)$ for the action of receiving a request to report
back what the reply would be to the request to process command $m$, and
$s(v)$ for the action of sending reply $v$.
\end{example}

\section{Concluding Remarks}
\label{sect-conclusions}

We have added the empty process constant to the different extensions of
\ACP\ with conditional expressions presented in~\cite{BM05a}.
In the past, the addition of the empty process constant to \ACP\ was
rather problematic.
Its current addition to the different extensions of \ACP\ with
conditional expressions presented in~\cite{BM05a} turns out to present
no additional complications.

The addition of the empty process constant to different extensions of
\ACP\ in this paper is based on the treatment of the empty process
constant in the setting of \ACP\ that is chosen in~\cite{BG87c}.
If it was based on the treatment of the empty process constant chosen
in~\cite{Vra97a} instead, the addition of the empty process constant to
different extensions of \ACP\ in this paper would have been slightly
different.
For example, with the treatment from~\cite{BG87c}, no special additional
axioms concerning conditional expressions are needed when adding the
empty process constant, whereas with the treatment from~\cite{Vra97a},
the special additional axiom $\ep \leftm (\phi \gc \ep) = \phi \gc \ep$
is needed.

In~\cite{BM05c}, we showed that threads, as found in programming
languages such as Java and C\#, and services used by them can be viewed
as processes that are definable over \ACPc, and that thread-service
composition on those processes can be expressed in terms of operators of
\ACPc\ extended with action renaming.
In fact, the termination behaviour of the composition of a thread with
the services used by it can be dealt with more directly, and without
action renaming, in \ACPec.

\bibliographystyle{splncs03}
\bibliography{PA}


\end{document}